\documentclass[lettersize,journal]{IEEEtran}
\usepackage{amsmath,amsfonts}
\usepackage{algorithm}
\usepackage{algpseudocode}
\usepackage{array}
\usepackage{textcomp}
\usepackage{stfloats}
\usepackage{url}
\usepackage{verbatim}
\usepackage{graphicx}
\usepackage{cite}
\usepackage{amssymb}
\usepackage{bm}
\usepackage{amsthm}
\usepackage[thinc]{esdiff}
\usepackage{float}
\usepackage{subcaption}
\usepackage{graphicx}
\usepackage{xcolor}
\newtheorem{theorem}{Theorem}
\usepackage{acronym}  
\acrodef{gs}[GC]{Ground Controller}
\acrodef{per}[PER]{Packet Error Rate}
\acrodef{sdr}[SDR]{Software-defined Radio}
\acrodef{sic}[SIC]{Successive Interference Cancellation}
\acrodef{iid}[i.i.d.]{independent and identically distributed}
\acrodef{adsb}[ADS-B]{Automatic Dependent Surveillance–Broadcast}
\acrodef{gps}[GPS]{Global Positioning System}
\acrodef{imm}[IMM]{Interacting Multiple Model}
\acrodef{psd}[PSD]{Power Spectral Density}
\acrodef{pmf}[PMF]{Probability Mass Function}
\acrodef{ppm}[PPM]{Pulse Position Modulation}
\acrodef{cfo}[CFO]{Carrier Frequency Offset}
\acrodef{awgn}[AWGN]{Additive White Gaussian Noise}
\acrodef{pdf}[PDF]{Probability Density Function}
\acrodef{atc}[ATC]{Air Traffic Control}
\acrodef{ads}[ADS]{Automatic Dependent Surveillance}
\acrodef{adsc}[ADS-C]{ADS-Contract}
\acrodef{kld}[KLD]{Kullback–Leibler Divergence}
\acrodef{poe}[POE]{Phase Offset Estimation}
\acrodef{po}[PO]{Phase Offset}
\acrodef{gm}[GM]{Gaussian Mixture}
\acrodef{mle}[MLE]{Maximum Likelihood Estimation}
\acrodef{em}[EM]{Expectation–Maximization}
\acrodef{ls}[LS]{Least Squares}
\acrodef{mse}[MSE]{Mean Squared Error}
\acrodef{snr}[SNR]{Signal-to-noise Ratio}
\acrodef{snrs}[SNRs]{Signal-to-Noise Ratios}
\acrodef{sop}[SoOp]{Signal-of-Opportunity}
\acrodef{imu}[IMU]{Inertial Measurement Unit}
\acrodef{gps}[GPS]{Global Positioning System}
\acrodef{ml}[ML]{Maximum Likelihood}

\begin{document}

\title{Joint Ranging and Phase Offset Estimation for Multiple Drones using ADS-B Signatures}

\author{
    Mostafa~Mohammadkarimi,~\IEEEmembership{Member,~IEEE},
    Geert Leus,~\IEEEmembership{Fellow,~IEEE}, \\ and
    Raj~Thilak~Rajan,~\IEEEmembership{Member,~IEEE}

%

\thanks{The authors are with the Faculty of Electrical Engineering,
Mathematics and Computer Science, Delft University of Technology, 2628 CD
Delft, The Netherlands (e-mail: m.mohammadkarimi@tudelft.nl, G.J.T.Leus@tudelft.nl, R.T.Rajan@tudelft.nl).
}
\thanks{This work is partially funded by the European Leadership Joint Undertaking (ECSEL JU), under grant agreement No 876019, and the ADACORSA project - ``Airborne Data Collection on Resilient System Architectures.” (https://adacorsa.eu/). }

\thanks{This paper published in IEEE Transactions on Vehicular Technology (DOI:10.1109/TVT.2023.3318192).}
}



\maketitle

\begin{abstract}
A new method for joint ranging and \ac{po} estimation of multiple drones/aircrafts is proposed in this paper.
The proposed method employs the superimposed uncoordinated \ac{adsb} packets broadcasted by drones/aircrafts for joint range and \ac{po} estimation.
It jointly estimates range and \ac{po} prior to ADS-B packet decoding; thus, it can improve air safety when packet decoding is infeasible due to packet collision. Moreover, it enables coherent detection of \ac{adsb} packets, which can result in more reliable multiple target tracking in aviation systems using cooperative sensors for detect and avoid (DAA).
By minimizing the \ac{kld} statistical distance measure, we show that the received complex baseband signal coming from $K$ uncoordinated drones/aircrafts corrupted by \ac{awgn} at a single antenna receiver can be approximated by an \ac{iid} \ac{gm} with $2^K$ mixture components in the two-dimensional (2D) plane. While direct joint \ac{mle} of range and \ac{po} from the derived \ac{gm} \ac{pdf} leads to an intractable maximization, our proposed method employs the \ac{em} algorithm to estimate the modes of the 2D Gaussian mixture followed by a reordering estimation technique through combinatorial optimization to estimate range and \ac{po}.
An extension to a multiple antenna receiver is also investigated in this paper. While the proposed estimator can estimate the range of multiple drones/aircrafts with a single receive antenna, a larger number of drones/aircrafts can be supported with higher accuracy by
the use of multiple antennas at the receiver.
The effectiveness of the proposed estimator is supported by simulation results. We show that the proposed estimator can jointly estimate the range of multiple drones/aircrafts accurately.
\end{abstract}


\acresetall		

\section{Introduction}
\IEEEPARstart{A}{utomatic}
Dependent Surveillance–Broadcast (ADS-B)
is one of the two Automatic Dependent Surveillance (ADS) systems that tracks aircrafts without employing radar.
It is intended to improve traffic surveillance capabilities by sharing accurate aircraft position information between pilots and air traffic controllers.
In the ADS-B system, aircrafts regularly and asynchronously broadcast their real-time position information, velocity, and identification to no specific receiver using a transponder, typically combined with a \ac{gps}, to transmit highly accurate positional information
to traffic controllers and directly to other aircrafts. This transmission is known as ADS-B Out and its accuracy and update rate
are much greater than conventional primary radar surveillance.
The reception of the ADS-B packet
 by an aircraft is ADS-B In \cite{strohmeier2014realities,tohidi2020compressed,kim2017ads}.


ADS-B system is considered a promising solution to enable safe autonomous drone navigation, especially in urban environments \cite{angelov2012sense}.
In order to avoid aviation accidents, each drone needs
to be aware of the position and speed of the surrounding
drones so that it can keep a safe separation distance with
the other drones. This safe separation can be achieved by a cooperative sensor system, such as the ADS-B system.  In this solution, drones are equipped with \ac{gps}, an \ac{imu}, and a miniaturized transponder and they broadcast their real-time position information, which can be employed by the surrounding drones or Ground Controllers (GCs) to maintain a safe operation distance of drones at low altitude and congested airspace.
In addition to ADS-B system,
 cooperative navigation by using
Wi-Fi and other  industrial, scientific, and medical (ISM) band wireless technologies   have been suggested  in \cite{minucci2020avoiding} and \cite{vinogradov2020wireless}. However, these solutions do not support long ranges.

Typically, ADS-B system is susceptible to severe
message collisions in dense air spaces. The
random channel access of the communication
protocols using the $1090$ MHz frequency
leads to ADS-B packet error rates above
50 percent for typical air space densities as
observed during the day \cite{strohmeier2014realities}.
One of the main challenges in the employment of a cooperative sensor system, such as ADS-B, for future drone technology is packet collisions due to a larger number of drones compared to aircrafts in the airspace. As the number of drones in the airspace increases, the probability of packet collision also increases.
The ADS-B system in its current form cannot handle packet collision; hence, a large number of packets are lost. Packet loss means less information and more uncertainty for the surrounding drones, resulting in less air safety.

Information extraction from collided and overlapping ADS-B packets, such as range, velocity,  Angle of Arrival (AoA), etc. can contribute to safer navigation of a drone/aircraft. These information can improve situational awareness and safety of the detect and avoid (DAA) systems.
Specifically, joint ranging and \ac{po} estimation enables coherent detection of a single or (collided) multiple ADS-B packets with a significantly lower \ac{per} compared to non-coherent detection.
To the best of the authors' knowledge, joint estimation of range and \ac{po} for multiple drones/aircrafts using collided ADS-B packets has not been investigated yet.
Existing blind source separation algorithms suffer from sign ambiguity and require the number of receive antennas not to be larger than the number of drones/aircrafts \cite{luo2018comprehensive,chabriel2014joint,naik2011overview}.

In this paper, we address the problem of joint range and \ac{po} estimation of multiple drones/aircrafts using the collided and overlapping ADS-B packets.
We analytically derive the \ac{ml} cost function for joint range and \ac{po} estimation of multiple drones/aircrafts. Then, a simple solution based on the \ac{em} algorithm is proposed.
Our proposed estimator enables avionic systems employing cooperative sensors to obtain range information of the surrounding drones/aircrafts while ADS-B packet decoding is not feasible due to collision.
Furthermore, the estimated range and \ac{po} can be employed for coherent detection of ADS-B packets, which offers higher detection performance. Our proposed estimator can
estimate the range of multiple drones/aircrafts with a single
receive antenna.

\subsection{Related Works}
Existing solutions for the separation of the overlapping ADS-B packets can be broadly divided into time-domain and spatial-domain methods \cite{li2021reliable}.
The spatial-domain methods take the advantage of antenna array and if the
direction of arrival of the aircrafts/drones signal are known,  the signal subspace methods, such as,  MUSIC \cite{schmidt1986multiple}, ESPRIT \cite{roy1989esprit}, and minimum variance distortionless response (MVDR) \cite{tzafri2016high}, can be employed for
ADS-B signal separation.
Projection
Algorithm (PA) and corresponding extensions for the separation of the secondary surveillance radar (SSR) signal have been proposed in \cite{petrochilos2004algorithms,petrochilos2005separation}.
ADS-B signal separation using high-order statistics of the received signal has shown to be
ineffective because the overlapping signal is pseudo-Gaussian \cite{petrochilos2009separation}. Furthermore,
Alternating Direction Method of
Multipliers (ADMM) has been suggested to solve the non-convex blind adaptive beamforming
problem for ADS-B signal separation in \cite{wang2019ads}.
The authors in \cite{petrochilos2007algebraic} showed that the performance of the blind source separation algorithms, such as, independent component
analysis (ICA) is not acceptable for ADS-B signal separation because of its short length.
Principal Component Analysis (PCA)  and Fast ICA algorithms for ADS-B signal separation have been proposed in
\cite{zhang2018optimization} and \cite{leonardi2019degarbling}.
A promising solution for ADS-B signal separation is Manchester decoding algorithm; however, as the delay between the reception of two ADS-B packets decreases, the performance significantly drops \cite{petrochilos2007algebraic}.

There are several works dealing with time-domain ADS-B signal separation.
The authors in \cite{dan2019single} have proposed to employ the empirical mode decomposition and ICA for the separation of overlapping ADS-B signals via single receive antenna.
The $K$-means clustering
for ADS-B signal separation based on
empirical mode decomposition and ICA has been developed in \cite{wu2017method}.
An anomaly doubt degree has been introduced in \cite{sunquan2018separation} to calculate signal overlap
time delay,
and adaptive threshold method
based on power difference has been employed to separate ADS-B signal when
the overlap signal is relatively large.
The problem of ADS-B signal separation using deep learning with a single receive antenna has been investigated in \cite{bi2022multi,wang2022single}.
The authors have shown that the separation accuracy of the deep learning based algorithm is higher than that of the traditional algorithms.
To the best of
the authors' knowledge,
the main issue with the above mentioned methods is that most of them can only separate two overlapping ADS-B signals.

\begin{figure}
\centering
\begin{subfigure}{0.45\textwidth}
    \includegraphics[width=\textwidth]{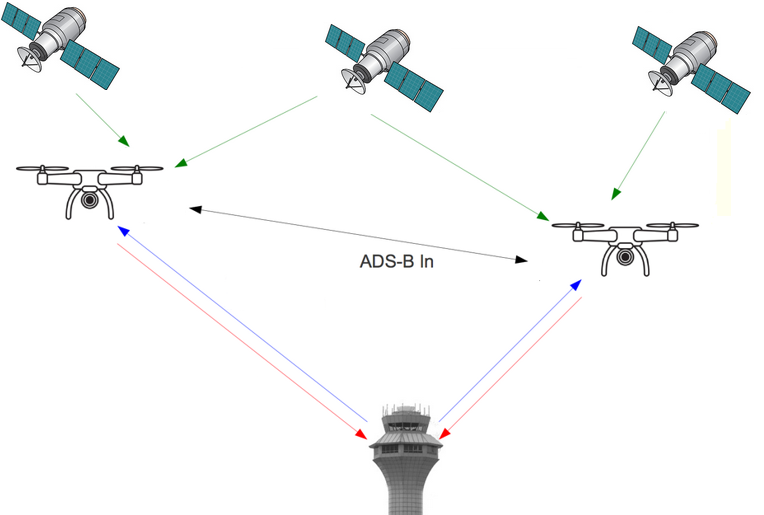}
    \caption{Range estimation in a \ac{gs}.}
    \label{fig:first}
\end{subfigure}
\\
\begin{subfigure}{0.45\textwidth}
    \includegraphics[width=\textwidth]{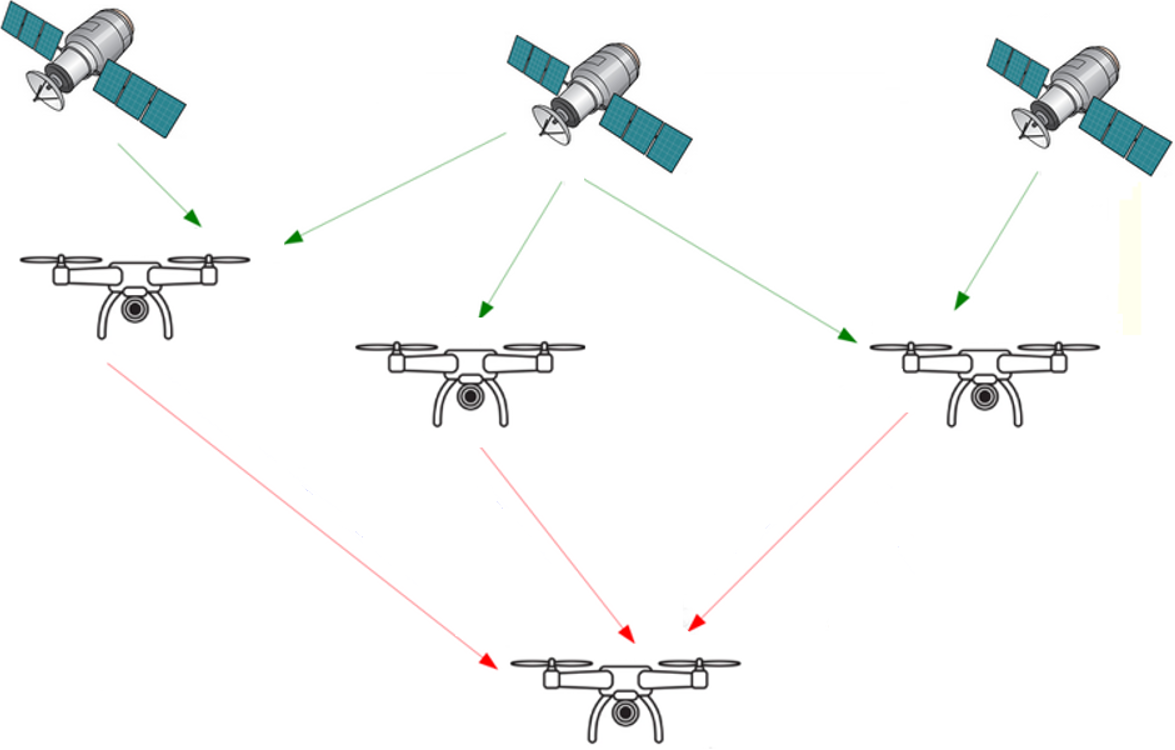}
    \caption{Range estimation in a flying drone.}
    \label{fig:second}
\end{subfigure}
\caption{Range estimation using the asynchronous ADS-B In signatures of the drones at the receiver.}
\label{fig:figures2}
\end{figure}

\subsection{Contributions}
Consider a collection of $K$ drones/aircrafts, which are asynchronously broadcasting ADS-B packets.
This simple transmission
protocol results in inevitable overlapping among multiple
ADS-B signals.
 In this work, we show that the received complex baseband signal from these drones/aircrafts can be approximated by an \ac{iid} \ac{gm} random variable with $2^K$ mixture components in the 2D plane, which are independent
of the arrival time of the ADS-B packets at the receiver.
Furthermore, by using the approximate \ac{pdf}, we derive the \ac{ml} cost function for the joint ranging and \ac{po} estimation of multiple drones/aircrafts.
We also propose the low-complexity \ac{em}-based joint ranging and \ac{po} estimation algorithm for
multiple drones/aircrafts. Our proposed estimator makes active multiple target ranging possible with a single receive antenna in the presence of ADS-B packet collision. Moreover, our solution enables ADS-B systems to coherent multi-packet decoding.
Finally, we extend the proposed \ac{em}-based joint estimator to the case of multiple receive antennas.

\subsection{Notations}
The identity matrix, all-zero vector, and all-one vector of length $N$ are denoted by ${\bf{I}}_N$, ${\bf{0}}_N$, and ${\bf{1}}_N$, respectively.
Throughout the paper, $(\cdot)^*$, $(\cdot)^{T}$, and $(\cdot)^{\rm{H}}$ show the complex conjugate, transpose, and Hermitian transpose, respectively. Also,
$| \cdot |$,  $\lfloor \cdot \rfloor$,   $*$,
 and $\otimes$  represent the absolute value operator, the floor function (greatest integer value),  linear convolution, and Kronecker product, respectively.
$\mathbb{E}\{\cdot\}$ is the statistical expectation, $\hat{{x}}$ is an estimate of $x$.
The complex Gaussian distribution with mean vector $\bm{\mu}$ and covariance matrix $\bm{\Sigma}$ is denoted by ${\cal{C}\cal{N}}\big{(}\bm{\mu},\bf{\Sigma}\big{)}$. The continuous uniform distribution between $a$ and $b$ and the discrete uniform distribution between $N_1$ and $N_2$ are denoted by $\mathcal{U}_{\rm c}[a,b]$ and $\mathcal{U}_{\rm d}[N_1,N_2]$, respectively.

The remainder of the paper is organized as follows. Section
\ref{system_model} introduces the system model. Section \ref{sec_III} describes the \ac{gm} distribution approximation by minimizing the
\ac{kld}.
In Section \ref{section_IV}, the maximum likelihood cost function and the \ac{em}-based joint ranging and \ac{po} estimation algorithm are analytically derived.
Reordering estimation for the proposed \ac{em}-based joint estimator through permutation-based combinatorial optimization
is investigated in Section \ref{Estimation_Map}. Joint estimation by taking the advantage of diversity gain through multiple receive antennas is discussed in Section \ref{mimo}.
Simulation results
are provided in Section \ref{simulation}, and conclusions are drawn in
Section \ref{conclusion}.


\begin{figure}[!t]
\centering
\includegraphics[width=3.5in]{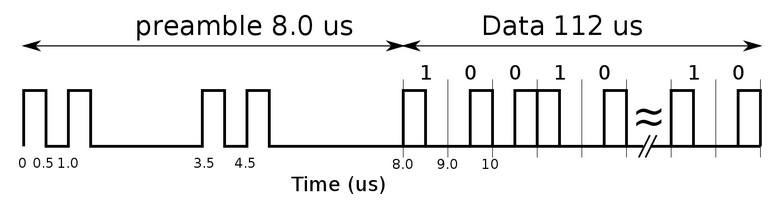}
\caption{An ADS-B packet is composed of a preamble and data in the \ac{ppm} form.}\label{Fig_2}
\end{figure}

\section{System Model}\label{system_model}
We consider $K$ drones broadcasting ADS-B packets through their
transponders to the \ac{gs}\footnote{The \ac{gs} can be a simple \ac{sdr} receiver. Number of drones/aircrafts, $K$, can be determined by employing model-order selection techniques \cite{shi2007adaptive,mohammadkarimi2017number} prior to joint ranging and \ac{po} estimation. } and also directly to other flying drones/aircrafts. Typical scenarios for signal reception at the \ac{gs} and the flying drone are shown in Fig.~\ref{fig:figures2}.
The packets are transmitted at $1090$ MHz and use \ac{ppm} at a rate of $1$ Mbit per second.


\begin{figure}
\centering
\begin{subfigure}{0.45\textwidth}
    \includegraphics[width=\textwidth]{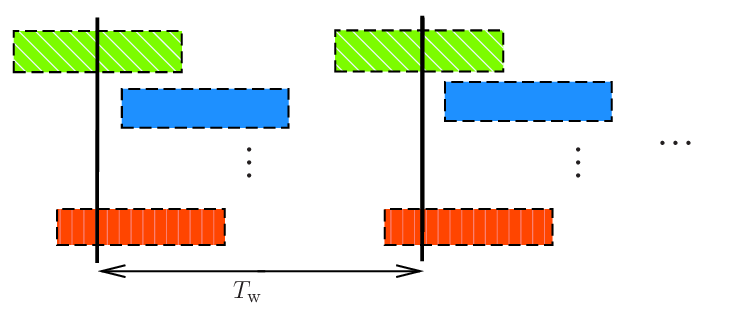}
    \caption{{The received ADS-B packet at the receiver and the observation window with length $T_{\rm w}$.}}
    \label{fig:first}
\end{subfigure}
\hfill
\begin{subfigure}{0.42\textwidth}
    \includegraphics[width=\textwidth]{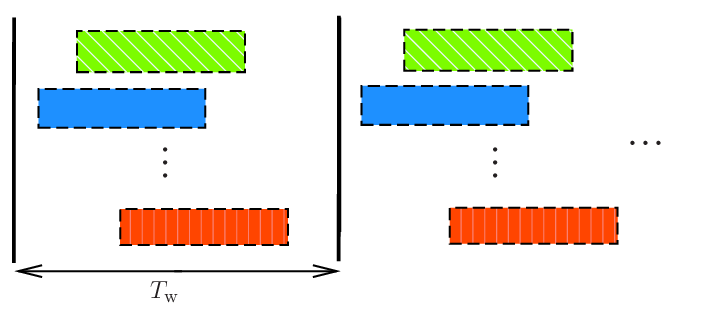}
    \caption{{A special case where the complete ADS-B packets of all the drones fall inside the observation window with length $T_{\rm w}$}.}
    \label{fig:second}
\end{subfigure}
\caption{{The reception of the ADS-B packets at the receiver. Drones periodically broadcast ADS-B packets.
Different colors are used to show the packet of drones.
}}
\label{fig:figures}
\end{figure}



It is assumed that the drones/aircrafts asynchronously broadcast their ADS-B packets every $T_{\rm P}$ seconds.
We consider an observation window of length $T_{\rm{w}}=T_{\rm P}$ for parameter estimation at the receiver (Fig.~\ref{fig:first}). To make the joint ranging and \ac{po} estimation independent of the arrival time of the ADS-B packets at the receiver, we approximate the received samples in the observation interval by an \ac{iid} complex random variable as it will be explained in Section \ref{sec_III}.
Hence,  without loss of generality,  we can consider the ADS-B packet reception in Fig.~\ref{fig:second} to simplify  modeling of the joint range and PO estimation.
In this case,
it is assumed that the ADS-B packet of the $k$th drone/aircraft with a packet length of $T_{\rm{A}}$
is received at the receiver with time delay $\tau_k \in [0,  \tau_{\rm max}]$ in the timing reference of the receiver,
where $\tau_k$ is unknown and random in each observation window of length $T_{\rm{w}}=T_{\rm P}$, and
$\tau_{\rm max}=T_{\rm P}-T_{\rm A}$ is the maximum time delay of a packet.

By employing a  baseband low pass filter with sufficient bandwidth $B$ at the receiver after RF down-conversion, the received complex baseband signal at the ADS-B receiver
 can be expressed as
\begin{align}\label{Observation}
y(t) \approx \sum_{k=1}^{K} \sqrt{P_k L_k}  x_k(t-\tau_k) e^{j (2 \pi  \Delta f_k t+\theta_k)} +w(t) ,
\end{align}
where $t \in [0 , T_{\rm w}]$, and where $P_k$, $x_k(t)$, $w(t)$, and $L_k$ denote, the transmit power by the $k$th drone,
the transmit \ac{ppm} waveform by the $k$th drone, the additive noise with \ac{psd} $N_0$ over the frequency $f \in [-B,B]$ , and the path loss between the $k$th drone/aircraft and the ADS-B receiver, respectively.
For free-space path loss, we have
\begin{align}\label{loss_range}
L_k \triangleq \Big{(}\frac{\lambda_{\rm{c}}}{4 \pi r_k}\Big{)}^2,
\end{align}
where $r_k$ is the range between the $k$th drone and the receiver,
$\lambda_{\rm{c}} \triangleq c / f_{\rm c}$ is the wavelength of the carrier wave, $c$ denotes the speed of light, and $f_{\rm c}$ represents the carrier frequency.
 For the ADS-B system, since $f_{\rm c}=1090$MHz, we have $\lambda_{\rm{c}} \approx 0.2752 $ m.
In \eqref{Observation}, $\Delta f_k$ and $\theta_k$ further denote the \ac{cfo} and the electrical \ac{po} of the $k$th drone in the observation window.
The \ac{cfo} and \ac{po} occur because the local oscillator signal for RF down-conversion at the receiver does not synchronize with the carrier signal.\footnote{The tolerable PO estimation error for coherent ADS-B detection depends on the value of SNR. The higher value of SNR, the less sensitivity to the PO estimation error in coherently detecting an ADS-B packet.}

In this paper, we consider that  $\Delta f_k T_{\rm w} \ll 1$ and $\exp(j 2 \pi  \Delta f_k t) \approx 1$ for $t \in [0,T_{\rm w}]$, $k=1,2,\cdots,K$; thus, we can write
\begin{align}\label{App_observation}
y(t) \approx \sum_{k=1}^{K} \sqrt{P_k L_k}  x_k(t-\tau_k)e^{j\theta_k}  +w(t).
\end{align}
{\bf Assumption 1:} We assume that $P_k$, $k=1,2,\ldots,K$, is known at the receiver, and that $P_1 L_1>  P_2 L_2 > \ldots, P_K L_K$.

Since the bit-rate for ADS-B systems is $1 {\rm M \ bit} / {\rm s}$,  a sampling rate of $f_s=\frac{1}{T_{\rm{s}}}=2 {\rm  M \ samples/s}$ is sufficient to capture the bit transitions, detect the preamble and decode packets. A typical sampled ADS-B signal is shown in Fig.~\ref{Fig_3} for visualization.
The discrete-time received baseband signal after sampling, i.e., $y_n \triangleq
y(nT_{\rm{s}})$, $n=0,1,\ldots,N$, where $N=239+M$,\footnote{A sampling rate of $2 {\rm  M \ samples/s}$ results in $240$ samples per ADS-B packet.}  $M  \triangleq \lfloor \frac{\tau_{\rm max}}{T_{\rm s}}\rfloor$, can be written in vector form as
\begin{align}\label{main_eq}
{\bf y} = \sum_{k=1}^{K} h_k{\bf x}_k + {\bf w} = \sum_{k=1}^{ K} {\bf z}_k + {\bf w}= {\bf g}+ {\bf w} ,
\end{align}
where ${\bf g} \triangleq \sum_{k=1}^{K} h_k{\bf x}_k=\sum_{k=1}^{K} {\bf z}_k$, ${\bf z}_k=h_k{\bf x}_k$,
\begin{subequations}
\begin{align}
{\bf y} &\triangleq    \big{[} y_0 \  y_1 \ \dots \ y_{N} \big{]}^T, \\
{\bf g} &\triangleq    \big{[} g_0 \  g_1 \ \dots \ g_{N} \big{]}^T, \\
{\bf w} &\triangleq    \big{[} w_0 \  w_1 \ \dots \ w_{N} \big{]}^T, \\
{\bf z}_k & \triangleq  \big{[} z_{k,0} \ z_{k,1} \ \cdots \ z_{k,N} \big{]}^T,
\end{align}
\end{subequations}
$w_n \triangleq
w(nT_{\rm{s}})$, $h_k \triangleq \beta_k e^{j\theta_k}$, $\beta_k \triangleq \sqrt{P_kL_k}$, and
\begin{align} \label{Vector}
{\bf x}_k & \triangleq  \big{[} x_{k,0} \ x_{k,1} \ \cdots \ x_{k,N} \big{]}^T \\
&\triangleq \nonumber
 \big{[} {\bf{0}}_{m_k}^T \ {\bf s}^T \ {\bf d}_k^T \ {\bf{0}}_{M-{m_k}}^T \big{]}^T.
\end{align}

In \eqref{Vector}, $x_{k,n} \triangleq x_k(nT_{\rm{s}}-\tau_k)$,
${\bf d}_k \in \{0,1\}$, is the data vector of the $k$th drone with a length of  $224$ symbols, ${\bf s}$ is the preamble vector given by
$
{\bf s} = \big{[} 1 \ 0  \ 1 \ 0 \  0 \ 0 \ 0 \  0 \ 0 \ 1 \ 0  \ 1 \ 0 \  0 \ 0 \ 0 \big{]}^T,
$ and
$M  \triangleq \lfloor \frac{\tau_{\rm max}}{T_{\rm s}}\rfloor$
and $m_k \triangleq \lfloor \frac{\tau_k}{T_{\rm s}}\rfloor$ denote the maximum possible integer delay for a drone/aircraft and
the integer delay of the $k$th drone/aircraft, respectively. The integer delays $m_k$, $k=1,2,\cdots,K$, are unknown at the receiver and their values change from one observation window to another.
The vector ${\bf w} \triangleq [w_0 \ w_1 \ \ldots \ w_N]^T$ in \eqref{main_eq} denotes
the additive white Gaussian noise (AWGN) with
covariance matrix
$\mathbb{E} \big{\{}{\bf w} {\bf w}^T\big{\}}=\sigma_{\rm w}^2 {\bf I}=2N_0B {\bf I}$.

Let us define hypothesis $H_m^k$ as follows:
\begin{align}\label{As}
H_m^k: \,\,\,\,\,\ {\bf x}_k = \big{[} {\bf{0}}_{m}^T,{\bf s}^T, {\bf d}_k^T , {\bf{0}}_{M-{m}}^T \big{]}^T,
\end{align}
which represents the ADS-B packet of the $k$th drone arriving at the receiver with integer delay $m_k=m\in \{0,1,\ldots,M\}$.

\begin{figure}[!t]
\centering
\includegraphics[width=3.5in]{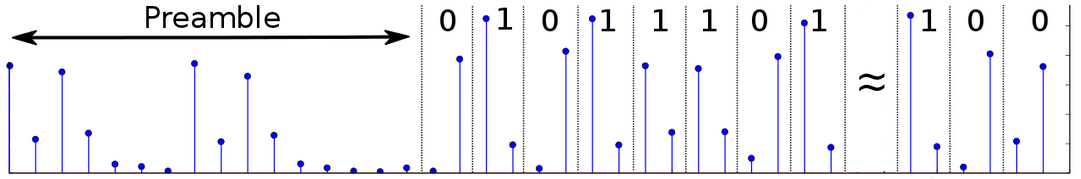}
\caption{Amplitude of the noisy ADS-B packet after sampling.}\label{Fig_3}
\end{figure}

\section{Distribution Approximation}\label{sec_III}
To remove the dependency of joint ranging and \ac{po} estimation from the unknown arrival time of the ADS-B packets at the receiver, we approximate the received noisy samples in the observation interval by \ac{iid} complex random variables.  In this paper, we show that an \ac{iid} two-dimensional (2D) \ac{gm} model can be used to model the received baseband superimposed and noise corrupted signal at the receiver.


\begin{theorem}\label{Theorem_1}
By maximizing the KLD criterion, the elements of the ADS-B packet of the $k$th drone, i.e., ${\bf x}_{k} =[
{x}_{k,0} \ {x}_{k,1} \ \ldots \ {x}_{k,N}]^T= \big{[} {\bf{0}}_{m_k}^T \ {\bf s}^T \ {\bf d}_k^T \ {\bf{0}}_{M-{m_k}}^T \big{]}^T$
can be approximated by an \ac{iid} random variable  that are Bernoulli distributed with \ac{pmf}

\begin{align}\label{Bernoulli}
q\big{(}x;p\big{)} = \begin{cases}
   p & \text{if }x=0, \\
    1-p & \text {if } x = 1,
 \end{cases}
\end{align}
where
\begin{align}\label{Bernoilli_Parameter}
p = \frac{M+124}{M+240}.
\end{align}
\end{theorem}
\begin{proof}
See Appendix \ref{Appendix_1}.
\end{proof}
As seen in Theorem \ref{Theorem_1}, the approximated \ac{pmf} $q$ does not depend on $H_m^{k}$.

Since ${\bf z}_{k} = h_{k} {\bf x}_{k}$
is the scaled version of ${\bf x}_{k}$ and the parameter of the Bernoulli distribution, $p$, is independent of $m_k$, the elements of the complex vector ${\bf z}_{k}$ can be approximated by \ac{iid} complex random variables $Z_k$ with \ac{pdf} $f_{{Z}_{k}}\big{(}z;p,h_k\big{)}$ as follows
\begin{align}\label{Theory_1}
f_{{Z}_{k}}\big{(}z;p,h_k\big{)}
=p\delta_{\rm c}(z)+(1-p)\delta_{\rm c}(z-h_k),
\end{align}
where $\delta_{\rm c}(z)$, $z=z_{\rm r}+jz_{\rm I} \in {\mathbb C}$, is the complex Delta function and is defined as
\begin{align}
\delta_{\rm c}(z) \triangleq \delta(z_{\rm r})\delta(z_{\rm I}),
\end{align}
where $\delta(t)$, $t \in {\mathbb R}$,  is the Dirac Delta function.

We consider that
$g_i \triangleq \sum_{k=1}^{K}z_{k,i}   \sim G$ and $z_{k,i} \sim Z_k$ given in \eqref{Theory_1}, where the symbol $\sim$ denotes distributed according to.
The \ac{pdf} of the sum of independent random variables is obtained as the convolution of the \ac{pdf}s.
For the complex random variable, $G=\sum_{k=1}^{K} Z_k$, by employing the multi-binomial theorem \cite{morris1975central}, we can obtain the  \ac{pdf} of
$G$ as
\begin{align} \nonumber
 f_{G}&(g;p,{\bf{h}})  =
\sum_{v_1=0}^{1} \cdots \sum_{v_K=0}^{1}
\bigg{[} p^{\sum_{k=1}^{K} v_k} (1-p)^{K-{\sum_{k=1}^{K} v_k}} \\ \label{delta}
& \,\,\,\,\,\,\,\,\,\,\,\,\,\,\,\,\,\,\,\,\,\,\,\,\,\,\,\,\,\,\,\ \times  \delta_{\rm c} \Big{(}g-\sum_{k=1}^{K}(1-v_{k})h_{k}\Big{)} \bigg{]},
\end{align}
where ${\bf{h}}\triangleq [h_1,h_2,\cdots,h_K]^T$.

\begin{proof}
See Appendix \ref{Appx_B}.
\end{proof}

For the circularly symmetric complex Gaussian noise vector, ${\bf w} \triangleq    \big{[} w_0 \  w_1 \ \dots \ w_{N} \big{]}^T$,
the \ac{pdf} of the random variable $W$ associated with the noise elements is expressed as
\begin{align}\label{noise_pdf}
f_{{ W}} ({w};\sigma_{\rm w}^2) \triangleq \frac{1}{{\pi \sigma_{\rm{w}}^2}} \exp{\bigg{(}\frac{-|w|^2}{\sigma_{\rm{w}}^2}\bigg{)}},
\end{align}
where $w\in \mathbb{C}$.
From \eqref{main_eq}, we have $y_n = g_n+w_n$, $n=0,1,\dots,N$, where $g_n \sim G$ and $w_n \sim W$.
Since $G$ and $W$ are independent complex random variables, the \ac{pdf} of $Y=G+W$ is obtained by the linear convolution of  the \ac{pdf}s in
\eqref{delta} and \eqref{noise_pdf}, which results in
\begin{align}\label{pdf_mixture}
&f_{Y}({{y}};p,{\bm \beta},{\bm \theta},\sigma_{\rm{w}}^2) = \sum_{a=0}^{2^K-1} \frac{\xi_a }{{\pi \sigma_{\rm{w}}^2}}  \mathcal{C}\mathcal{N}(y;\mu_a,\sigma_{\rm{w}}^2)
\\ \nonumber
& =\sum_{a=0}^{2^K-1} \frac{\xi_a }{{\pi \sigma_{\rm{w}}^2}}  \exp{\Bigg{(}-\frac{|y-\mu_a|^2}{\sigma_{\rm{w}}^2}\Bigg{)}},
\end{align}
where
\begin{subequations}
\begin{align}
{\bm \beta} & \triangleq \Big{[}\beta_1 \ \beta_2\ \ldots \ \beta_{K} \Big{]}^T,
\\
{\bm \theta} & \triangleq \big{[}\theta_1 \ \theta_2 \ \ldots \ \theta_K \big{]}^T,
\end{align}
\end{subequations}
\begin{align}\label{pi}
\xi_a &\triangleq p^{\sum_{k=1}^{K} b_k} (1-p)^{K-{\sum_{k=1}^{K} b_k}},
\end{align}
and
\begin{align}\label{tyu}
{{\mu}}_a  \triangleq  \sum_{k=1}^{K}  (1-b_k)h_{k}
=\sum_{k=1}^{K}  (1-b_k)\beta_{k}\exp(j\theta_{k}).
\end{align}
with $b_i$ the $i$th bit in the binary representation of $a$ as
\begin{align}
a&=(b_K,b_{K-1},\ldots,b_1)_2, \,\,\,\,\,\,\,\,\,\,\,\,\,\ b_i \in \{0,1\},
\end{align}
and $a= 0 ,1 , \dots ,2^K-1$.

As seen, $f_{Y}({{y}};p,{\bm \beta},{\bm \theta},\sigma_{\rm{w}}^2)$ represents a 2D \ac{gm}, where its modes are located at the delta functions given in \eqref{delta}. Since $\theta_1,\theta_2,\ldots,\theta_K$, are independent uniform random variables in the range of $[0,2\pi)$, and  $\beta_1>\beta_2>\ldots > \beta_K$, with probability of almost one, the number of distinct mixtures is $2^K$.
Fig.~\ref{Fig_x} illustrates the constellation of the received signal for $K=3$ drones, maximum integer delay $M=20$, and transmit power of $51$ dBm.

\begin{figure}[!t]
\centering
\includegraphics[width=3.2in]{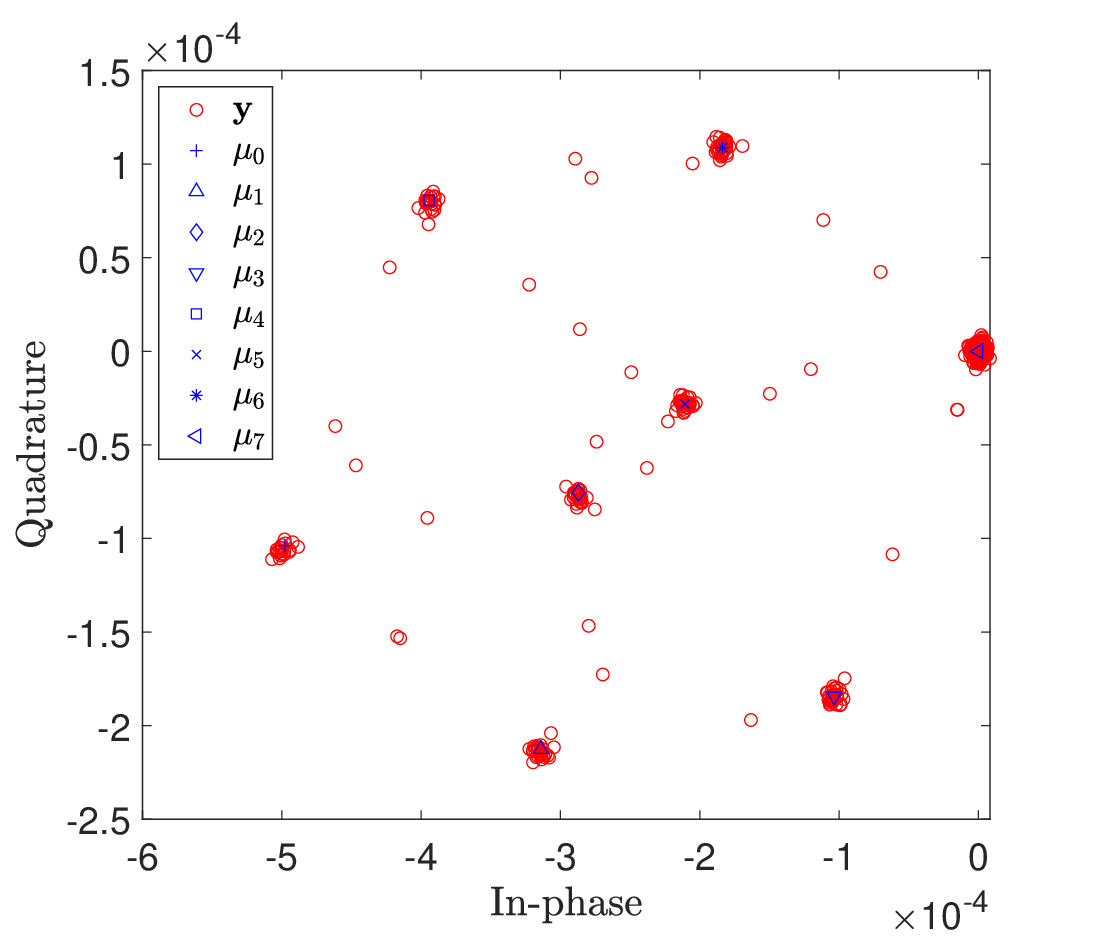}
\caption{{The in-phase and quadrature components of the received signal for $K=3$ drones at $r_1=r_2=r_3=12$ Km with transmit power of $51$ dBm.}}\label{Fig_x}
\end{figure}

\section{Joint Range and \ac{po} Estimation}\label{section_IV}
The \ac{mle} for the  vector parameters $[{\bm{\beta}}^T \ \bm{\theta}^T]^T$ given observation vector ${\bf{y}} =[{y}_0 \ {y}_2 \ \ldots \ {y}_{N}]^T$ is expressed as
\begin{align}\label{mle}
\{\hat{\bm{\beta}},\hat{\bm{\theta}}\}= \operatorname*{arg\,max}_{{\bm \beta},{\bm{\theta}}} \sum_{n=0}^{N} \ln f_{{Y}}({ y}_n;p,{\bm \beta },\bm{\theta}, \sigma_{\rm{w}}^2).
\end{align}
The maximization problem in \eqref{mle} cannot be analytically solved in a trackable manner. An alternative simple solution is to employ the \ac{em} algorithm to estimate the $2^K$ modes of the
\ac{gm} components; then, we can decouple the desired parameters, i.e., ${\bm \beta}$ and ${\bm \theta}$ from the estimated modes.

Let $\bm{\mu} \triangleq [\mu_0 \ \mu_1 \ \ldots \ \mu_{2^K-1}]^T$ denote the mode vector of the $\ac{gm}$, where $\mu_a$, $a=0,1,\ldots,2^K-1$, is given by
\eqref{tyu}. {The elements of the vector $\bm{\mu}$ for $K=2,3,4$ are given in Appendix \ref{Appendix 3}.}

Let us define the discrete function $\chi_q(n)$ as
\begin{align}\label{fun}
\chi_q(n): \{1,2,\ldots,l\}\longrightarrow \{1,2,\ldots,l\},
\end{align}
where for $n_1 \neq n_2$, $\chi_q(n_1) \neq \chi_q(n_2)$. There are $Q_l \triangleq  l!$ unique functions in the form of \eqref{fun}, where $!$ denotes the factorial function.
Using \eqref{fun}, $Q_l$ permutation matrices of size $l \times l$ can be defined as
\begin{align}\label{permu}
{{\bf \Lambda}_q} = \left[ \begin{array}{l}
{{\bf e}_{{\chi _q}(1)}}\\
{{\bf e}_{{\chi _q}(2)}}\\
\,\,\,\, \vdots \\
{{\bf e}_{{\chi _q}({l})}}
\end{array} \right],
\end{align}
where ${\bf{e}}_{\ell}$, $\ell=1,2,\ldots,l$,  denote the standard basis vectors of length $l$ with a $1$ in the $\ell$th coordinate and 0's elsewhere.
The set composed of all permutations of vector ${\bf a} =[a_1 \ a_2 \ \ldots \ a_l]^T$ is given by
\begin{align}
\mathcal{F}_{\bf a}^{Q_l} \triangleq \Big{\{} {{{\bf{\Lambda }}_1}{\bf{a }},{{\bf{\Lambda }}_2}{\bf{a }}, \ldots ,{{\bf{\Lambda }}_{Q_l}}{\bf{a }}} \Big{\}},
\end{align}
where $Q_l =  l!$.

The \ac{em} algorithm estimates the permuted mode vector ${\bm \eta} =[\eta_0 \ \eta_1 \ \ldots  \eta_{2^K-1}]^T \in \mathcal{F}_{\bm \mu}^{Q_{2^K}} \subset  \mathbb{C}^{2^K}$,
where $Q_{2^K}=2^K!$ and
${\bm{\mu}} \triangleq [{{\mu}}_0 \ {{\mu}}_1 \ \ldots \ {{\mu}}_{2^K-1}]^T$.
The \ac{em} algorithm
defines a latent random vector ${\bf u}\triangleq [u_0\ u_1 \ \ldots \ u_{N}]^T$ that determines the \ac{gm} component from which the observation originates, i.e., $f_{{Y}|U}({{y}}_n|u_n=a;p,{\bm \beta},{\bm \theta},\sigma_{\rm{w}}^2)
\sim \mathcal{C}\mathcal{N}(y_n;\mu_a,\sigma_{\rm{w}}^2)$,
where $P_U(u_n=a)=\xi_a$ for $n=0,1,\ldots,N$ and $a=0,1,\ldots,2^K-1$.
The \ac{em} algorithm iteratively maximizes the expected value of the complete-data log-likelihood function to estimate the permuted mode vector
${\bm \eta} =[\eta_0 \ \eta_1 \ \ldots  \eta_{2^K-1}]^T \in \mathcal{F}_{\bm \mu}^{Q_{2^K}} \subset  \mathbb{C}^{2^K}$
of the \ac{gm} as follows \cite{mclachlan2019finite,theodoridis2015machine,moon1996expectation}
\begin{align}\label{maximize}
\hat{\bm \eta}^{(t+1)} = \operatorname*{arg\,max}_{\bm \eta} Q({\bm \eta}|{\bm \eta}^{(t)}),
\end{align}
where ${\bm \eta}^{(0)}$ is the initialization vector,
\begin{align}\label{EM_main}
&Q({\bm \eta}|{\bm \eta}^{(t)}) = \mathbb{E}_{{\bm U}|{\bm Y},{\bm{\eta}}^{(t)}}\Big{\{}\ln f_{{\bm Y},{\bm U}}({\bf y},{\bf u};p,{\bm \eta},\sigma_{\rm{w}}^2)\Big{\}} \\ \nonumber
& = \sum_{n=0}^{N} \sum_{a=0}^{2^K-1} \lambda_{a,n}^{(t)}
 \Bigg{(} \ln \frac{\xi_a}{\pi \sigma_{\rm w}^2}  - \frac{|y_n-\eta_a|^2}{\sigma_{\rm{w}}^2}  \Bigg{)},
\end{align}
with
\begin{align}
\lambda_{a,n}^{(t)} &= P_{U|Y}\big{(}u_n=a|{y}_n;{\bm{\eta}}^{(t)}\big{)} \\ \nonumber
&=\frac{P_U(u_n=a) f_{Y|U}\big{(}{y}_n|u_n=a;{\bm{\eta}}^{(t)}\big{)} }{f_Y(y_n)} \\ \nonumber
&=\frac{\xi_a \mathcal{C}\mathcal{N}\big{(} y_n;\eta_a^{(t)},\sigma_{\rm{w}}^2)}{\sum_{q=0}^{2^K-1}\xi_q \mathcal{C}\mathcal{N}( y_n;\eta_q^{(t)},\sigma_{\rm{w}}^2)},
\end{align}
and the complete-data likelihood function is given by
\begin{align}\label{likelihood_1}
f_{{\bm Y, \bm U}}({\bf y},{\bf u};p,\bm{\eta},\sigma_{\rm{w}}^2)=
\prod_{n=0}^{N} \prod_{a=0}^{2^K-1}
\big{(} \xi_a \mathcal{C}\mathcal{N}(y_n;\eta_a,\sigma_{\rm{w}}^2)\big{)}^{\mathbb{I}\{u_n=a\}}.
\end{align}
In \eqref{likelihood_1}, $\mathbb{I}\{\cdot\}$ denotes the indicator function, and  $\xi_a$ is a function of $p$ and is given in \eqref{pi}.
The \ac{em} algorithm at the $(t+1)$th iteration estimates the vector ${\bm \eta}^{(t+1)} =[\eta_0^{(t+1)} \ \eta_1^{(t+1)} \ \ldots \eta_{2^K-1}^{(t+1)}]^T$ which is
a permuted version of the vector $\bm{\mu}$. The order of ${\bm \eta}^{(t+1)}$ depends on the initialization of the EM algorithm, i.e., ${\bm \eta}^{(0)}$. 
By solving the maximization problem in \eqref{maximize},
the elements of ${\bm \eta}^{(t+1)}$ are iteratively updated as follows
\begin{align}\label{mode_update_1}
\eta_a^{(t+1)}=\frac{\sum_{n=0}^{N}  \lambda_{a,n}^{(t)}y_n}{\sum_{n=0}^{N}  \lambda_{a,n}^{(t)}},
\end{align}
for $a=0,1,\dots, 2^K-1$, where
the convergence condition for the \ac{em} algorithm is $\|{\bm \eta}^{(t+1)}-{\bm \eta}^{(t)}\|<\epsilon$, with $\epsilon$ a preset threshold. We denote $\hat{\bm \eta}={\bm \eta}^{(t+1)}$ when the \ac{em} algorithm converges at the $(t+1)$th iteration.

The \ac{em} algorithm may converge to a local maximum of the observed data likelihood function, depending on the starting value. A variety of heuristic or metaheuristic approaches exist to escape a local maximum, such as random-restart hill climbing  where the \ac{em} algorithm starts with different random initial estimates. The recursive expression in \eqref{mode_update_1} is very straightforward and can be easily computed for multiple initial points. Hence, \ac{em} with random-restart hill climbing can be employed for the problem at hand.
As mentioned earlier, there is an order ambiguity in the vector $\hat{\bm \eta}$; thus, reordering estimation is required to resolve this ambiguity. In the next section, we propose different solutions based on  permutation-
based combinatorial optimization for one to one mapping between the elements of $\hat{\bm \eta}$ and ${\bm \mu}$.

\section{Reordering Estimation}\label{Estimation_Map}
The goal of reordering estimation is to change the order of the elements in
$\hat{\bm{\eta}} \triangleq [\hat{{\eta}}_0 \ \hat{{\eta}}_1 \ \ldots \ \hat{{\eta}}_{2^K-1}]^T$, obtained by the \ac{em} algorithm, to achieve a new vector $\hat{\bm{\mu}} \triangleq [\hat{{\mu}}_0 \ \hat{{\mu}}_1 \ \ldots \ \hat{{\mu}}_{2^K-1}]^T$ that corresponds to an estimate of ${\bm{\mu}} \triangleq [{{\mu}}_0 \ {{\mu}}_1 \ \ldots \ {{\mu}}_{2^K-1}]^T$.
Let us denote $\hat{{\eta}}_i$ as the element of $\hat{\bm{\eta}}$ that corresponds to ${\mu}_{2^K-1}=0$. The index $i$ can be estimated as
\begin{align}\label{first_elm}
|\hat{\eta}_i|<  \min \Big{\{}|\hat{\eta}_0|,|\hat{\eta}_1|,\ldots,|\hat{\eta}_{i-1}|,|\hat{\eta}_{i+1}|,\ldots,|\hat{\eta}_{2^K-1}|\Big{\}}.
\end{align}
Accordingly, we have $\hat{{\mu}}_{2^K-1} = \hat{\eta}_i$.
%
Let us now define
\begin{align}\label{set_omit}
{\hat{\bm{\eta}}}_i \triangleq [\hat{{\eta}}_0 \ \hat{{\eta}}_1 \ \ldots \ \hat{\eta}_{i-1} \  {\hat \eta}_{i+1} \ \ldots \ \hat{{\eta}}_{2^K-1}]^T,
\end{align}
and
\begin{align}\nonumber
{\hat{\mathcal{A}}}_{l} & \triangleq \Big{\{}[\phi_1 \ \phi_2 \ \ldots \ \phi_l]^T\big{|} \forall {d} \in \{1,2,\ldots, l \},
\\  \label{set}
& \,\,\,\,\,\,\,\,\,\,\ \phi_d \in \{\hat{{\eta}}_0 , \hat{{\eta}}_1 , \ldots \ \hat{\eta}_{i-1} ,  {\hat \eta}_{i+1} \ \ldots \ \hat{{\eta}}_{2^K-1}\}, \\ \nonumber
& \,\,\,\,\,\,\,\,\,\,\ |\phi_1|> |\phi_2| > \ldots |\phi_l|
\Big{\}}.
\end{align}

 Different reordering estimation methods for $\hat{\bm \beta}$ and $\hat{\bm \theta}$ can be considered.
The \ac{ls} reordering estimation for $\hat{\bm \beta}$ and $\hat{\bm \theta}$  form $\hat{\bm{\eta}}_i$
is given by
\begin{align}\label{estimate}
\hat{\bf h} = \hat{\bm \beta}e^{j\hat{\bm \theta}}= \hat{\bf \Lambda}{\bf A}\hat{\bf \Phi},
\end{align}
where
\begin{align}\label{eq:ls}
\{\hat{\bf \Lambda},\hat{{\bf \Phi}}\} &= \operatorname*{arg\,min}_{{\bf \Lambda},{{\bf \Phi}}} \big{\|}{\bf \Lambda}{\bf A}{\bf \Phi}- \hat{\bm{\eta}}_i\big{\|}_2, \\ \nonumber
&\textrm{s.t.} \,\ {\bf \Phi} \triangleq [\Phi_1,\Phi_2,\ldots,\Phi_K]^T \in  {\hat{\mathcal{A}}}_{K}\\ \nonumber
&\,\,\,\,\,\,\,\ {\bf \Lambda} \in \{{\bf \Lambda}_1, {\bf \Lambda}_2, \ldots , {\bf \Lambda}_{(2^{K-1})!}\}
\end{align}
where ${\bf \Lambda}_i$ is a permutation matrix of size $2^{K-1} \times 2^{K-1}$ given in \eqref{permu} for $l=2^{K-1}$, and ${\bf A}$ is the $2^{K-1}\times K$ matrix given by
\begin{align}
{\bf A} \triangleq
\left[ \begin{array}{l}
\,\,\,\,\,\,\,\,\,\,\,\,\,\,\,\,\,\,\,\,\ {{\bf{e}}_1}\\
\,\,\,\,\,\,\,\,\,\,\,\,\,\,\,\,\,\,\,\,\ {{\bf{e}}_2}\\
\,\,\,\,\,\,\,\,\,\,\,\,\,\,\,\,\,\,\,\,\,\ \vdots \\
\,\,\,\,\,\,\,\,\,\,\,\,\,\,\,\,\,\,\,\,\ {{\bf{e}}_K}\\
\,\,\,\,\,\,\,\,\,\,\,\,\,\,\,\,\,{{\bf{e}}_1} + {{\bf{e}}_2}\\
\,\,\,\,\,\,\,\,\,\,\,\,\,\,\,\,\,{{\bf{e}}_1} + {{\bf{e}}_3}\\
\,\,\,\,\,\,\,\,\,\,\,\,\,\,\,\,\,\,\,\,\,\, \vdots \\
{{\bf{e}}_1} + {{\bf{e}}_2} +  \ldots  + {{\bf{e}}_K}
\end{array} \right],
\end{align}
where ${\bf{e}}_j$, $j=1,2,\ldots,K$,  denote the standard basis vectors of length $K$.

An alternative minimization problem with lower computational complexity for vector reordering is given by

\begin{align}\label{eq:ls2}
\{\hat{\bf \Lambda},\hat{{\bf \Phi}}\} &= \operatorname*{arg\,min}_{{\bf \Lambda},{{\bf \Phi}}} \big{\|}{\bf \Lambda}{\bf A}{\bf \Phi}- \hat{\bm{\eta}}_i\big{\|}_2, \\ \nonumber
&\textrm{s.t.} \,\ {\bf \Phi} \triangleq [\Phi_1,\Phi_2,\ldots,\Phi_K]^T \in  \mathbb{C}^K \\ \nonumber
&\,\,\,\,\,\,\,\ {\bf \Lambda} \in \{{\bf \Lambda}_1, {\bf \Lambda}_2, \ldots , {\bf \Lambda}_{(2^{K-1})!}\}
\end{align}
By employing the solution of the unconstrained \ac{ls} minimization for a linear observation model \cite{kay1993fundamentals} and the fact that ${\bf \Lambda}^T{\bf \Lambda}={\bf \Lambda}{\bf \Lambda}^T={\bf I}_{2^K-1}$, we can easily write
\begin{align}
{{\bf \Phi}} = ({\bf A}^T{\bf A})^{-1}{\bf A}^T{\bf \Lambda}^T\hat{\bm{\eta}}_i,
\end{align}
and thus we can formulate the following problem
\begin{align}\label{eq:ls3}
\hat{\bf \Lambda} &= \operatorname*{arg\,min}_{{\bf \Lambda}} \hat{\bm{\eta}}_i^T (\hat{\bm{\eta}}_i-{\bf \Lambda}{\bf A}{{\bf \Phi}}) \\ \nonumber
&\textrm{s.t.} \,\ {{\bf \Phi}} = ({\bf A}^T{\bf A})^{-1}{\bf A}^T{\bf \Lambda}^T\hat{\bm{\eta}}_i \\ \nonumber
& \,\,\,\,\,\,\,\ |\Phi_1|> |\Phi_2|>\ldots > |\Phi_K| \\ \nonumber
&\,\,\,\,\,\,\,\ {\bf \Lambda} \in \{{\bf \Lambda}_1, {\bf \Lambda}_2, \ldots , {\bf \Lambda}_{(2^{K-1})!}\}.
\end{align}

Our simulation experiments show that for low and moderate \ac{snrs}, the vector reordering based on the combinatorial optimization in \eqref{eq:ls2} outperforms the one in \eqref{eq:ls} in terms of estimation error. While the vector reordering based on the minimization formulation in \eqref{eq:ls2} offers a lower computational complexity compared to the minimization in \eqref{eq:ls}, the computational complexity of both methods is still high for $K>4$.


An alternative low complexity solution is to define a combinatorial optimization with lower cardinality.
By selecting all $Q$-combination of the set $\{\hat{{\eta}}_0 \ \hat{{\eta}}_1 \ \ldots \ \hat{\eta}_{i-1} \ {\hat \eta}_{i+1} \ \ldots \ \hat{{\eta}}_{2^K-1}\}$, we can define a $C_Q^{2^K-1}$-cardinality combinatorial optimization problem that minimizes a linear/non-linear  combinations of the elements in the set in such a way that one or more elements of the set can be
unambiguously assigned to the elements of $\bm{\mu}$.

As an example, for $K=4$ drones/aircrafts, our goal is to estimate
${\mu}_{14}=\beta_1\exp(j\theta_1)$, ${\mu}_{13}=\beta_2\exp(j\theta_2)$,
${\mu}_{11}=\beta_3\exp(j\theta_3)$, and ${\mu}_{7}=\beta_4\exp(j\theta_4)$, where $\beta_1> \beta_2> \beta_3> \beta_4$ (refer to \eqref{4mu}).
We can easily show that there is a unique solution for the linear equation
\begin{align}
z_0+z_1+z_2+z_3-z_4=0,
\end{align} where
$z_i\in \{\mu_0 \ \mu_1 \ \ldots \ \mu_{14}\}$ and $|z_1|>|z_2|> |z_3|> |z_4|$. This unique solution is
${\mu}_{14}+{\mu}_{13}+{\mu}_{11}+{\mu}_{7}-{\mu}_{0}=0$, where
${\mu}_{0} =\sum_{i=1}^{4}\beta_i\exp(j\theta_i)$, ${\mu}_{14}=\beta_1\exp(j\theta_1)$, ${\mu}_{13}=\beta_2\exp(j\theta_2)$,
${\mu}_{11}=\beta_3\exp(j\theta_3)$, and ${\mu}_{7}=\beta_4\exp(j\theta_4)$ for $\beta_1> \beta_2> \beta_3> \beta_4$.
By taking this into account, we can define the following combinatorial optimization for reordering estimation
\begin{align}\label{eq:25s}
\hat{\bm \Phi} &= \operatorname*{arg\,min}_{\bm \Phi} \ \big{|}{{\bf v}^T}{\bm \Phi} \big{|}\\ \nonumber
&\textrm{s.t.} \,\ {\bm \Phi} \triangleq [\phi_1 \ \phi_2 \ \ldots \ \phi_5]^T   \in \hat{\mathcal A}_{5} \nonumber
\end{align}
where $\hat{\bm \Phi} \triangleq [\hat{\phi}_1 \ \hat{\phi_2} \ \ldots \ \hat{\phi_5}]^T$,
the set $\hat{\mathcal A}_{5}$ is defined in \eqref{set} for $l=5$, and ${\bf v} \triangleq [v_1 \ v_2 \ v_3 \ v_4 \ v_5]^T= [1 \ 1 \ 1 \ 1 \ -1]^T$.
It is obvious that \eqref{eq:25s} represents a $C_5^{15}=3003$-cardinality combinatorial optimization problem.
By solving the  combinatorial optimization in \eqref{eq:25s}, we obtain
\begin{align}
\hat{\bf h} = \hat{\bm \beta}e^{j\hat{\bm \theta}}=[\hat{\phi}_1 \ \hat{\phi_2} \ \hat{\phi}_3 \ \hat{\phi_4}]^T.
\end{align}

Different combinatorial optimization problems can be defined for reordering estimation.
For $K=4$, let us consider the following linear equation
\begin{align}\label{Comb_4}
7 \sum_{i=0}^{3} z_i -\sum_{i=4}^{14}  z_i=0,
\end{align}
 where
$z_i\in \{\mu_0 \ \mu_1 \ \ldots \ \mu_{14}\}$ and $|z_1|>|z_2|> |z_3|> |z_4|$.
Since $\beta_1>\beta_2>\beta_3>\beta_4$ (refer to assumption 1), we can show that the solution of \eqref{Comb_4} is given by $7\mu_{14}+7\mu_{13}+7\mu_{11}+7\mu_{7}-{\bf c}^T{\bf \Lambda}^T =0$, where
${\bf c}\triangleq[\mu_{0}\ \mu_1 \ \mu_2 \ \mu_3 \ \mu_4 \ \mu_5 \ \mu_6  \ \mu_8 \ \mu_9 \ \mu_{10} \ \mu_{12}]^T$ and
 ${\bf \Lambda}$ is a $11 \times 11$ permutation matrix.\footnote{The second summation is independent of the permutation of $\bf c$.} Taking this equation into account, the combinatorial optimization for reordering estimation can be expressed as
 \begin{align}\label{eq:ks}
\hat{\bm \Phi} &= \operatorname*{arg\,min}_{{\bm{\Phi}}} \ \Big{|}7{\bm \Phi}^T  {\bf 1}_4 -\alpha\Big{|}, \\ \nonumber
&\textrm{s.t.} \,\ {\bm \Phi} \triangleq [\phi_1 \ \phi_2 \ \phi_3 \ \phi_4]^T   \in \hat{\mathcal A}_{4} \\ \nonumber
&\,\,\,\,\,\,\,\ \alpha=\sum_{\kappa\in {\mathcal S}  }\kappa
\end{align}
%
where ${\mathcal S} \triangleq \{\hat{{\eta}}_0 , \hat{{\eta}}_1 , \ldots \ \hat{\eta}_{i-1} ,  {\hat \eta}_{i+1} \ \ldots \ \hat{{\eta}}_{2^K-1}\}-\{\phi_1,\phi_2,\phi_3,\phi_4\}$, $\hat{\bm \Phi}  \triangleq [\hat{\phi}_1 \ \hat{\phi}_2 \ \hat{\phi}_3 \ \hat{\phi}_4]^T$, ${\bf 1}_4 =  [1 \ 1 \ 1 \ 1]^T$, and
the set $\hat{\mathcal A}_{4}$ is defined in \eqref{set} for $l=4$. By using the solution of the minimization in \eqref{eq:ks},  we can write $\hat{\bf h} = \hat{\bm \beta}e^{j\hat{\bm \theta}}=[\hat{\phi}_1 \ \hat{\phi_2} \ \hat{\phi}_3 \ \hat{\phi_4}]^T$.


It should be mentioned that other combinatorial optimizations can also be defined for reordering estimation.
Ambiguity removal through multiple combinatorial optimization problems can also be considered.

\subsection{Joint Range and \ac{po} Estimation}
For joint range and \ac{po} estimation of $K$ drones, the modes $\mu_{a_k}\triangleq \sqrt{P_{k}L_{k}}e^{j\theta_{k}}$, $k=1,2,\ldots,K$, are needed to be estimated, where
\begin{align}\label{ak}
a_k = \sum_{\substack{n=0 \\ n \neq k-1}}^{K-1} 2^n.
\end{align}
Let $\hat{\mu}_{a_k}$, $k=1,2,\ldots,K$, denote the estimated modes after reordering estimation.
By using \eqref{loss_range}, the range and \ac{po} for the $k$th drone are estimated as
\begin{align}\label{range_es}
{\hat{r}}_{k}= \frac{\lambda_{\rm c}\sqrt{P_{k}}}{\big{|}4\pi\hat{\mu}_{a_{k}}\big{|}},
\end{align}
and

\begin{align}\label{phase_es}
\hat{\theta}_{k} =
\begin{cases}
\tan^{-1}\frac{\Im \{ \hat{\mu}_{a_k} \}}{\Re \{\hat{\mu}_{a_k}\}},\,\,\,\,\,\,\,\,\,\,\,\,\,\,\,\,\,\,\ {\Re \{\hat{\mu}_{a_k}\}}\geq 0 \\
\tan^{-1}\frac{\Im \{ \hat{\mu}_{a_k} \}}{\Re \{\hat{\mu}_{a_k}\}}+\pi, \,\,\,\,\,\,\,\,\ {\Re \{\hat{\mu}_{a_k}\}}< 0
\end{cases}.
\end{align}
where $\Im\{\cdot\}$ and $\Re\{\cdot\}$ denotes the real and imaginary operators, respectively.
The proposed \ac{em}-based joint ranging and \ac{po} estimation is
summarized in Algorithm \ref{euclid}, where $T_{\rm EM}$ denotes the number of iterations of the \ac{em} algorithm.\footnote{The pseudocode has been written for the reordering methods in \eqref{eq:ls} and \eqref{eq:ls2}. Similar pseudocode can be written for other reordering estimation methods.}

  \begin{algorithm}[!t]
    \caption{: \ac{em}-based joint ranging and \ac{po} estimation}\label{euclid}
    \begin{algorithmic}[1]
    \Statex \textbf{Input:} $P_1,P_2,\ldots,P_K$ and $\bf y$
    \Statex \textbf{Output:} $\hat{r}_k$ and $\hat{\theta}_k$ for $k=1,2,\ldots,K$
     \Statex \text{Initialize ${\bm{\eta}}^{(0)}$ }
 \For {$t=0,1,\ldots,T_{\rm EM}-1$}
\State  \text{Use \eqref{mode_update_1} to obtain ${\bm{\eta}}^{(t+1)} \triangleq [{{\eta}}_0^{(t+1)} \ {{\eta}}_1^{(t+1)} \ \ldots \ {{\eta}}_{2^K-1}^{(t+1)}]^T$}
\State {\bf if}   $\|{\bm \eta}^{(t+1)}-{\bm \eta}^{(t)}\|<\epsilon$ or $t=T_{\rm EM}-1$
\State $\hat{\bm \eta} = {\bm \eta}^{(t+1)}$ and $t=T_{\rm EM}$
\State {\bf end if}
\EndFor
\State Use $\eqref{first_elm}$ to estimate the element of $\hat{\bm{\eta}}$ that corresponds to ${\mu}_{2^K-1}=0$ and obtain ${\hat{\bm{\eta}}}_i$ in \eqref{set_omit}.
\State Solve the combinatorial optimization in \eqref{eq:ls} or \eqref{eq:ls2} to obtain $\{\hat{\bf \Lambda},\hat{{\bf \Phi}}\}$
\State Estimate  $\hat{\bf h}=[\hat{\mu}_{a_1} \ \hat{\mu}_{a_2} \ldots \ \hat{\mu}_{a_K}]$ by using \eqref{estimate} for $a_k$, $k=1,2,\ldots,K$, given in \eqref{ak}.
\State Estimate  $r_k$ and $\theta_k$ by employing \eqref{range_es} and \eqref{phase_es}.
    \end{algorithmic}\label{2mcmult}
  \end{algorithm}

\section{Multiple Antennas Receiver}\label{mimo}
In this section, we extend the derived maximum likelihood cost function and the proposed \ac{em}-based joint range and \ac{po} estimators to the case of multiple receive antennas.

\subsection{Maximum Likelihood and \ac{em} Estimators}
{We consider $N_{\rm r}$ single antenna receivers where we assume that the distance between the receive antenna elements is more than half a wavelength. Under this assumption, the coupling effect can be neglected and spatially uncorrelated Gaussian noise can be considered.
We propose a time-domain estimator and does not consider the directivity of the multiple receive antennas in contrast to an antenna array because of the random phase of each receiver. It should be mentioned that an antenna array  is a set of multiple connected antennas which work together as a single antenna to transmit or receive radio waves. However, we consider independent single antenna receivers that take the advantage of combining gain.
}

 With the assumption that the time delay between receive antennas is negligible, and
the path loss between the $k$th drone and  $\ell$th receive antenna is the same for all receive antennas, i.e,
$L_{\ell,k}=L_k$, $k=1,2,\ldots,K$, $\ell=1,2,\ldots,N_{\rm r}$,
the received complex baseband  signal at the multiple-receive antennas can be expressed as
\begin{align}\label{eq:observation}
{\bf Y} = {\bf H} {\bf X} + {\bf W},
\end{align}
where ${\bf X} \triangleq [{\bf x}_1 \ {\bf x}_2 \ \ldots \ {\bf x}_K]^T \in \mathbb{C}^{K \times (N+1)}$, ${\bf Y} \triangleq [{\bf y}_0 \  {\bf y}_1 \ldots \ {\bf y}_{N}]\in \mathbb{C}^{N_{\rm{r}}\times (N+1)}$, ${\bf x}_k$ is given by \eqref{Vector}, and ${\bf y}_{n} \triangleq [y_{1,n} \ y_{2,n} \ \ldots \ y_{N_{\rm r},n}]^T$ denotes the received vector at time index $n$. In \eqref{eq:observation},
the matrices
${\bf{H}}\in \mathbb{C}^{N_{\rm{r}}\times K}$ and ${\bf{W}}\in \mathbb{C}^{N_{\rm{r}}\times (N+1)}$ are given as
\begin{align}
{\bf{H}} = \left[ \begin{array}{l}
{\bf{h}}_1^T\\
{\bf{h}}_2^T\\
\,\, \vdots \\
{\bf{h}}_{{N_{\rm{r}}}}^T
\end{array} \right], \,\,\,\,\,\,\,\,\,\
{\bf{W}} = \left[ \begin{array}{l}
{\bf{w}}_1^T\\
{\bf{w}}_2^T\\
\,\, \vdots \\
{\bf{w}}_{{N_{\rm{r}}}}^T
\end{array} \right],
\end{align}
where
\begin{align}
{\bf{h}}_\ell & \triangleq \big{[}h_{\ell,1} \ h_{\ell,2} \ \ldots \ h_{\ell,K}\big{]}^T \\ \nonumber
& = \big{[}\beta_1 \exp(j\theta_{\ell,1}) \ \ldots \ \beta_N \exp(j\theta_{\ell,K})\big{]}^T,
\end{align}
$\beta_k=\beta_{\ell,k} = \sqrt{P_kL_k}$, ${\bf w}_{\ell} \triangleq [w_{\ell,0} \ w_{\ell,1} \ \ldots \ w_{\ell,N}]^T$ with $w_{\ell,n} \sim \mathcal{C} \mathcal{N}  (0,\sigma_{\rm{w}}^2)$ the complex Gaussian noise at the $\ell$th receive antenna at time index $n$.
As seen, while $\beta_{1,k}=\beta_{2,k}=\cdots=\beta_{N_{\rm{r}},k}=\beta_k$, the phases $\theta_{1,k}, \theta_{2,k}, \ldots, \theta_{N_{\rm{r}},k}$ are independent random values in $[0 \ 2 \pi)$.

The joint \ac{pdf} of the elements of ${\bf Y}$ is given by
\begin{align}\label{joint_pdf_ma}
f_{\bf Y}({\bf{Y}};p,{\bm \beta},{\bm \Theta},\sigma_{\rm{w}}^2) = \prod_{\ell=1}^{N_{\rm r}} \prod_{n=0}^{N}  f_{Y}({{y}_{\ell,n}};p,{\bm \beta},{\bm \theta}_\ell,\sigma_{\rm{w}}^2),
\end{align}
where ${\bm \beta} \triangleq [\beta_1 \ \beta_2 \ \ldots \ \beta_K]^T$, ${\bm \Theta} \triangleq \big{[}{\bm \theta}_1^T \ {\bm \theta}_2^T \ \ldots \ {\bm \theta}_{N_{\rm r}}^T \big{]}^T$, ${\bm \theta}_{\ell}  \triangleq \big{[}\theta_{\ell,1} \ \theta_{\ell,2} \ \ldots \ \theta_{\ell,K} \big{]}^T$, and $f_{Y}({{y}};p,{\bm \beta},{\bm \theta},\sigma_{\rm{w}}^2)$ is given in \eqref{pdf_mixture}. We can easily write the direct \ac{mle} of the parameter vector $\bm{\beta}$ and ${\bm \Theta}$ as
\begin{align}\label{mle_ma}
\{\hat{\bm \beta}, \hat{\bm \Theta} \} =  \operatorname*{arg\,max}_{{\bm \beta},{\bm{\Theta}}} \sum_{\ell=1}^{N_{\rm r}}   \sum_{n=0}^{N} \ln f_{{Y}}\big{(}{ y}_{\ell,n};p,{\bm \beta },\bm{\theta}_\ell, \sigma_{\rm{w}}^2\big{)}.
\end{align}
Solving the maximization problem in \eqref{mle_ma} is not straightforward. Hence, the direct \ac{mle}  cannot be obtained.
However, indirect estimation can be obtained  by estimating the modes of the \ac{gm} similar to the single receive antenna.

Analogous to the single receive antenna scenario, we can employ the \ac{em} algorithm for estimating the modes of the \ac{gm}.
For multiple receive antennas, the \ac{em} algorithm
defines an identical latent random vector ${\bf u}\triangleq [u_0 \ u_1 \ \ldots \ u_N]^T$ for all receive antennas. This random vector determines the \ac{gm} component from which the observation originates, i.e.,
\begin{align}
f_{Y}({{y}_{\ell,n}}|u_n=a;p,{\bm \beta},{\bm \theta}_\ell,\sigma_{\rm{w}}^2) \sim \mathcal{C}\mathcal{N}(y_{\ell,n};\mu_{\ell,a},\sigma_{\rm{w}}^2) ,
\end{align}
where $n=0,1,\ldots,N$, $\ell =1,2,\ldots, {N_{\rm{r}}}$, $a= 0 ,1 , \dots ,2^K-1$,
${{\mu}}_{\ell,a}   \triangleq
\sum_{k=1}^{K}  (1-b_{k})h_{\ell,k}= \sum_{k=1}^{K}  (1-b_{k})\beta_{k}\exp(j\theta_{\ell,k})$,
$a=(b_K,b_{K-1},\ldots,b_1)_2, b_i \in \{0,1\}$, and $P_{U}(u_n=a)=\xi_a$.

The \ac{em} algorithm iteratively maximizes the expected value of the complete-data log-likelihood function to estimate the vector
 ${\bm \Gamma} \triangleq [{\bm \eta}_1^T \ {\bm \eta}_2^T \ \ldots \ {\bm \eta}_{N_{\rm r}}^T]^T$ as

\begin{align}\label{EM_ma}
{\bm \Gamma}^{(t+1)} =  \operatorname*{arg\,max}_{{\bm \Gamma}}C\big{(}{\bm \Gamma}|{\bm \Gamma}^{(t)}\big{)},
\end{align}
where the vector ${\bm \eta}_{\ell} \triangleq [\eta_{\ell,0} \ \eta_{\ell,1} \ \ldots \ \eta_{\ell,2^K-1}]^T$ denotes the $2^K$ modes of the 2D \ac{gm} at the $l$th receive antenna,
${\bm \Gamma}^{(0)}$ is the initialization vector,
\begin{align}\label{Expectation_a}
&C({\bm \Gamma}|{\bm \Gamma}^{(t)}) = \mathbb{E}_{{\bm U}|{\bm Y};{\bm{\Gamma}}^{(t)}}\Big{\{}\ln f_{{\bm Y},{\bm U}}({\bf Y},{\bf u};p,{\bm \Gamma},\sigma_{\rm{w}}^2)\Big{\}} \\ \nonumber
&=  \sum_{\ell=1}^{N_{\rm r}}
\sum_{n=0}^{N} \sum_{a=0}^{2^K-1} \delta_{a, n}^{(t)}
 \Bigg{(} \ln \frac{\xi_a}{\pi \sigma_{\rm w}^2}   - \frac{|y_{\ell,n}-\eta_{\ell,n}|^2}{\sigma_{\rm{w}}^2} \Bigg{)},
\end{align}
\begin{align}
\delta_{a,n}^{(t)} &= P_{U|{{\bm Y}}}\big{(}u_n=a|{\bf{y}}_n;p,{\bm{\Gamma}}^{(t)},\sigma_{\rm w}^2\big{)} \\ \nonumber
&=\frac{P_U(u_n=a)f_{{\bm{Y}}|{U}}\big{(}{\bf{y}}_n|u_n=a;p,{\bm{\Gamma}}^{(t)},\sigma_{\rm{w}}^2\big{)}}{\sum_{a=0}^{2^K-1}f_{\bm Y,U}({\bf y}_n,u_n=a;p,{\bm{\Gamma}}^{(t)},\sigma_{\rm w}^2)} \\ \nonumber
&=\frac{\xi_a
\prod_{\ell=1}^{N_{\rm{r}}}
\mathcal{C}\mathcal{N}\big{(} y_{\ell,n};\eta_{\ell,a}^{(t)},\sigma_{\rm{w}}^2)}{\sum_{a=0}^{2^K-1}\xi_a \prod_{\ell=1}^{N_{\rm r}} \mathcal{C}\mathcal{N}( y_{\ell,n};\eta_{\ell,a}^{(t)},\sigma_{\rm{w}}^2)},
\end{align}
and the joint complete-data likelihood function is given by
\begin{align}\label{likelihood_MA}
f_{{\bm Y},{\bm U}}({\bf Y},&{\bf u};p,\bm{\Gamma},\sigma_{\rm{w}}^2) \\ \nonumber
&=
\prod_{\ell=1}^{N_{\rm{r}}}
\prod_{n=0}^{N} \prod_{a=0}^{2^K-1}
\big{(} \xi_a \mathcal{C}\mathcal{N}( y_{\ell,n};\mu_{\ell,a},\sigma_{\rm{w}}^2)\big{)}^{\mathbb{I}\{u_n=a\}}.
\end{align}

By solving the maximization problem in \eqref{EM_ma},
the elements of ${\bm \eta}^{(t+1)}$ are iteratively updated as follows
\begin{align}\label{mode_update}
\eta_{\ell,a}^{(t+1)}=\frac{\sum_{n=0}^{N}  \delta_{a,n}^{(t)}y_{\ell,n}}{\sum_{n=0}^{N}  \delta_{a,n}^{(t)}}.
\end{align}
The convergence condition for the \ac{em} algorithm is $\|{\bm \Upsilon}^{(t+1)}-{\bm \Upsilon}^{(t)}\|<\epsilon$, where $\epsilon$ is a preset threshold. We denote $\hat{\eta}_{\ell,a}={\eta}_{\ell,a}^{(t+1)}$ when the \ac{em} algorithm converges at the $(t+1)$th iteration.

While the \ac{em} algorithm estimates $N_{\rm r}2^K$ parameters, only $N_{\rm r}K$ parameters are used for joint ranging and \ac{po} estimation of $K$ drones. These $N_{\rm r}K$ modes are ${\mu}_{\ell,{a_1}}, {\mu}_{\ell,a_2}, \ldots , {\mu}_{\ell,a_{K}}$, where $a_k$ is defined in \eqref{ak}.
After the EM convergence, each receive antenna independently applies estimation mapping as explained in section \ref{Estimation_Map}.
Let $\hat{\mu}_{\ell,{a_1}}, \hat{\mu}_{\ell,a_2}, \ldots , \hat{\mu}_{\ell,a_{K}}$ denote the estimated and reordered modes
at the $\ell$th receive antenna. By averaging, then we can write
\begin{align}\label{EG2}
|\hat{\mu}_{{a_k}}|=\frac{1}{N_{\rm r}} \sum_{\ell=1}^{N_{\rm r}}|\hat{\mu}_{\ell,{a_k}}|\propto \frac{1}{\hat{r}_{k}},
\end{align}
where $k=1,2,\ldots,K$.
By substituting \eqref{EG2} into \eqref{range_es}, we can estimate the range of the $k$th drone.
The \ac{po} for each drone-receive antenna is obtained by replacing $\hat{\mu}_{\ell,{a_1}}, \hat{\mu}_{\ell,a_2}, \ldots , \hat{\mu}_{\ell,a_{K}}$, $\ell=1,2,\cdots,N_{\rm r}$ into \eqref{phase_es}.

It should be mentioned that the \ac{em} algorithm can also be implemented independently at each receive antennas. Then, after reordering estimation, the ranges are obtained by averaging. This results in a lower complexity solution.

\subsection{Outlier Detection}
For multiple receive antennas at the receiver,
it is possible that one or multiple outliers appear in the estimated sequence
 $|\hat{\mu}_{\ell,{a_1}}|, |\hat{\mu}_{\ell,a_1}|, \ldots , |\hat{\mu}_{\ell,a_{K}}|$, where $a_k$ is given in \eqref{ak}.
 These outliers can increase the estimation error of $|\hat{\mu}_{{a_k}}|$ in \eqref{EG2} that is used for range estimation.  To  remove the effect of outliers, we can use outlier detection algorithms, such as
the median absolute deviation (MAD) \cite{chandola2009anomaly}. In this case, the averaging is performed over the receive antennas without outliers. This results in a significant performance improvement in range estimation.

\subsection{Computational Complexity Analysis}
In each iteration of the \ac{em} algorithm, we evaluate $N_{\rm r} 2^K$ Gaussian densities for $N+1$ observation
points in the E-step in \eqref{Expectation_a}, and its computational cost scales as $O(N_{\rm r}N2^K)$. The computational cost of the M-step in \eqref{mode_update} per iteration also scales as $O(N_{\rm r}N2^K)$.
The complexity of the $k$-means$++$ initialization is $O(N_{\rm r}N2^K)$. The number of combinatorial search for reordering estimation via
the \ac{ls} and the proposed
linear/non-linear combinations minimization
 is $(2^K-1)!$ and $C_Q^{2^K-1}$, $Q=1,2,\ldots,2^K-1$, (depending on the selected linear/non-linear combinations), respectively. Hence, for $K \geq 4$, the method of linear/non-liner combinations minimization
 is more computationally efficient for reordering estimation.

\begin{table}\label{nonumber}
\centering
{
\caption{Operation parameters for the simulation \cite{series2017reception}.}
\vspace{1em}
\begin{tabular}{ |l|l| }
  \hline
  \multicolumn{2}{|c|}{ADS-B parameters for the simulation} \\
  \hline
  $P_{\rm t}=51$ dBm & ADS-B transmit power \\
  $f_{\rm c}=1090$ MHz & Carrier frequency of the ADS-B system \\
  $B_3=1.3$ MHz & 3 dB bandwidth of the transmit ADS-B signal \\
  $B_{20}=7$ MHz & 20 dB bandwidth of the transmit ADS-B signal \\
  $B_{40}=23$ MHz & 40 dB bandwidth of the transmit ADS-B signal \\
  $B_{60}=78$ MHz & 60 dB bandwidth of the transmit ADS-B signal \\
  $T_{\rm r}=0.01 \mu {\rm s}$ & Rise time of the ADS-B trapezoidal transmit pulse \\
  $T_{\rm d}=0.01 \mu {\rm s} $ & Decay time of the ADS-B trapezoidal transmit pulse \\
  $T=0.48 \mu {\rm s}$ & Time of the ADS-B trapezoidal transmit pulse \\
  \hline
\end{tabular}}\label{tabale:2}
\end{table}

\section{Simulation}\label{simulation}
In this section, we examine the performance of the proposed
\ac{em}-based joint range and \ac{po} estimator through several simulation experiments to confirm the effectiveness of the proposed algorithm.

\subsection{Simulation Setup}
Unless otherwise mentioned, we considered $K$ drones with ranges $r_1,r_2,\ldots,r_K \in \mathcal{U}_{\rm c}[1,10]$ Km.
The azimuth and elevation angles of the drones are assumed to be $\phi_1,\phi_2,\ldots, \phi_K \in \mathcal{U}_{\rm c}[-\pi,\pi)$ and $\psi_1,\psi_2,\ldots,\psi_K \in \mathcal{U}_{\rm c}[0,\pi/2]$, respectively.
We considered free space path loss model in \eqref{loss_range} and $\lambda_{\rm{c}} \approx 0.2752 $ m for ADS-B systems.
The number of receive antennas was set to $N_{\rm r}=5$.
%
%
%
For ADS-B message period of $T_{\rm P}=240 \mu {\rm s}$, the discrete time delay of each drone for each ADS-B packet transmission is assumed to be modeled by a discrete \ac{iid} random variable with uniform distribution as
$m_1,m_2,\ldots,m_K \in \mathcal{U}_{\rm d}[0,M]$, where $M \triangleq  \lfloor 2B T_{\rm p} \rfloor$ with $B$ as the bandwidth of the square-root-raised-cosine (SRRC) receive baseband filter $h(t)$ with roll-off factor $\beta=0.9$ and group delay $\tau_{\rm gr}=47.25 \mu {\rm s}$.
Without loss of generality, we consider that $E_{\rm h} \triangleq \int_{-\infty}^{+\infty} |h(t)|^2 {\rm d}t=\int_{-B}^{+B} |H(f)|^2 {\rm d}f =1$, where $H(f)$ is the frequency response of the receive filter.

We also assumed that the \ac{po} between  the $k$th drone and the $\ell$th receive antenna is modeled by a continuous \ac{iid} random variable with uniform distribution as $\theta_{\ell,k} \in \mathcal{U}_{\rm c}[0,2\pi)$, $k=1,2,\ldots,K$, $\ell=1,2,\ldots,N_{\rm r}$. Without loss of generality, we considered that $P_1=P_2=\dots=P_K=51$ dBm \cite{series2017reception}, and
the noise power in dBm at each receive antennas is considered to be $\sigma_{\rm w}^2 = 10\log_{10} ((N_0 E_{\rm h})/10^{-3})= -174$.


\begin{figure}[!t]
\vspace{0.3em}
\centering
\includegraphics[width=3.1in]{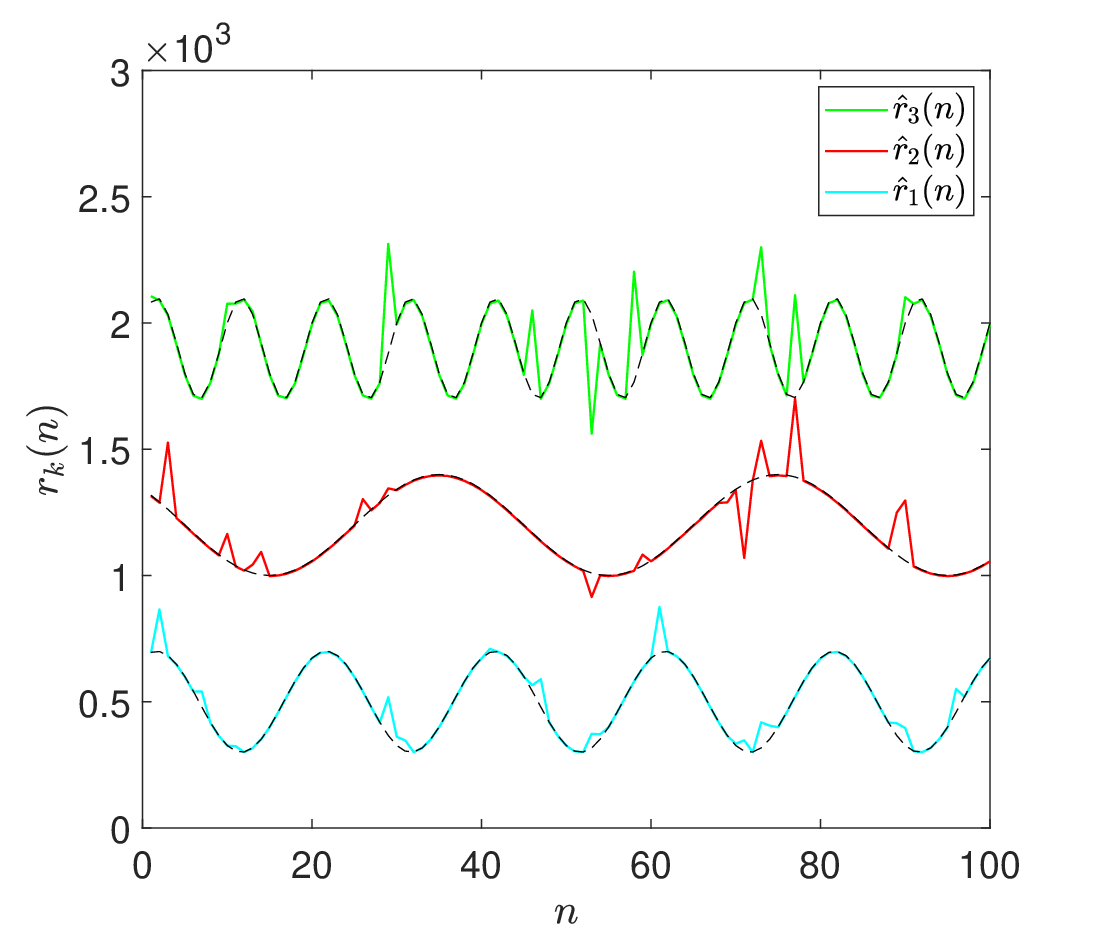}
\vspace{0.12em}
\caption{{Range tracking by the proposed \ac{em}-based estimator for $K=3$ drones. The dashed and solid lines denote the true and estimated range, respectively.} }\label{Range_Tracking}
\end{figure}

The performance of the proposed \ac{em}-based joint range and \ac{po} estimator was evaluated in terms of $1-P_{{\rm{out}},{r}}$ and $1-P_{{\rm{out},{\rm \theta}}}$
where
\begin{align}\nonumber
P_{{\rm{out}},r} \triangleq \frac{1}{K}\sum_{k=1}^{K} \mathbb{P}\bigg{\{}\frac{|\hat{r}_k-r_k|}{r_k} > \alpha_r \bigg{\}},
\end{align}
and
\begin{align}\nonumber
P_{{\rm{out}},\theta} \triangleq \frac{1}{KN_{\rm{r}}}\sum_{\ell=1}^{N_{\rm r}} \sum_{k=1}^{K} \mathbb{P}\bigg{\{}\frac{|\hat{\theta}_{\ell,k}-\theta_{\ell,k}|}{\theta_{\ell,k}} > \alpha_\theta \bigg{\}},
\end{align}
with $\hat{r}_k$ as the estimated range of the $k$th drone and $\hat{\theta}_{\ell,k}$ as the \ac{po} of the $k$th drone at the $\ell$th receive antenna.
The \ac{em} algorithm was run with $100$ different random initial values by using the $k$-means$++$ initialization,
and we considered the minimization problem in \eqref{eq:ls2} for reordering estimation.
 The number of Monte Carlo runs is  $10^4$.


\subsection{Simulation Results}
{To illustrate the performance of the proposed \ac{em}-based joint estimator over time,  we show range tracking for $K=3$ drones with transmit power $P_1=P_2=P_3=51$ dBm in Fig.~\ref{Range_Tracking}. It is considered that the ADS-B packets overlap.
The range, delay, and \ac{po} variations versus index of the ADS-B packet, $n$, for the three drones are modeled as
$r_1(n)= 500+200 \cos(0.1 \pi n+\pi /3)$, $r_2(n)= 1200+200 \cos(0.05 \pi n+2\pi/4)$,
$r_3(n)= 1900+200 \cos(0.2 \pi n+\pi/6)$, $m_k(n) \in \mathcal{U}_{\rm d}[0,17280]$,
$\theta_{\ell,k}(n) \in \mathcal{U}_{\rm c}[0,2\pi)$ for $\ell=1,2,\ldots,5$, $k=1,2,3$, $n=1,2,\ldots,100$, and $B=36$ MHz.  We consider that $\mathbb{E}\{\theta_{\ell_1,k_1}(n_1) \theta_{\ell_2,k_2}(n_2)\}=(\pi^2/3) \delta[\ell_1-\ell_2]
 \delta[k_1-k_2]\delta[n_1-n_2]$.
As seen, the proposed \ac{em}-based estimator can accurately track the range of $K=3$ drones. While the estimation error for the closest drone to the receiver is lower, i.e., $\mathbb{E}\{|{r}_1-\hat{r}_1|\}^2 < \mathbb{E}\{|{r}_2-\hat{r}_2|\}^2 < \mathbb{E}\{|{r}_3-\hat{r}_3|\}^2$,
$1-P_{{\rm{out}},r}$ is almost the same for all drones.}
%

\begin{figure*}[h!]
\centering
\subfloat[\small{Range estimation for $K=2$ drones}]{%
  \includegraphics[width=8cm]{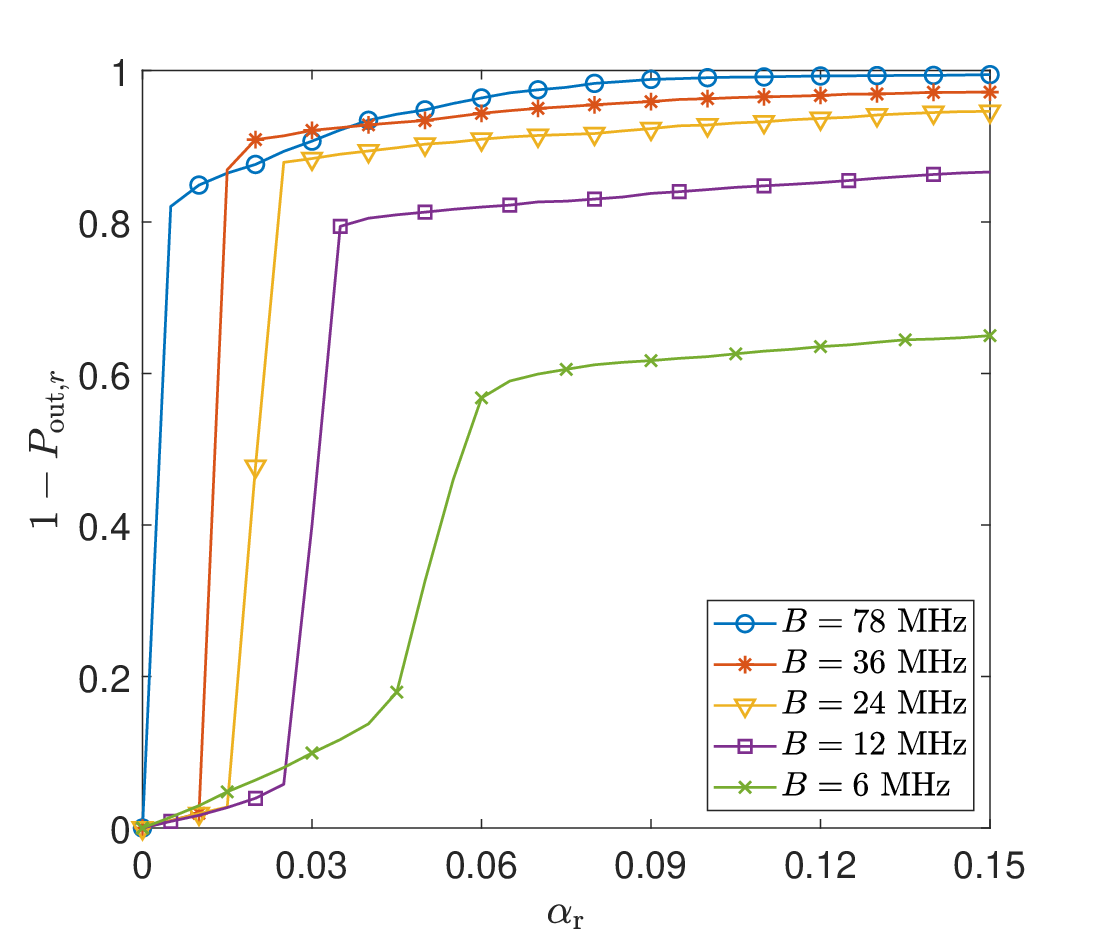}%
  \label{fig:evaluation:revenue}%
}\qquad
\subfloat[\small{\ac{po} estimation for $K=2$ drones}]{%
  \includegraphics[width=8cm]{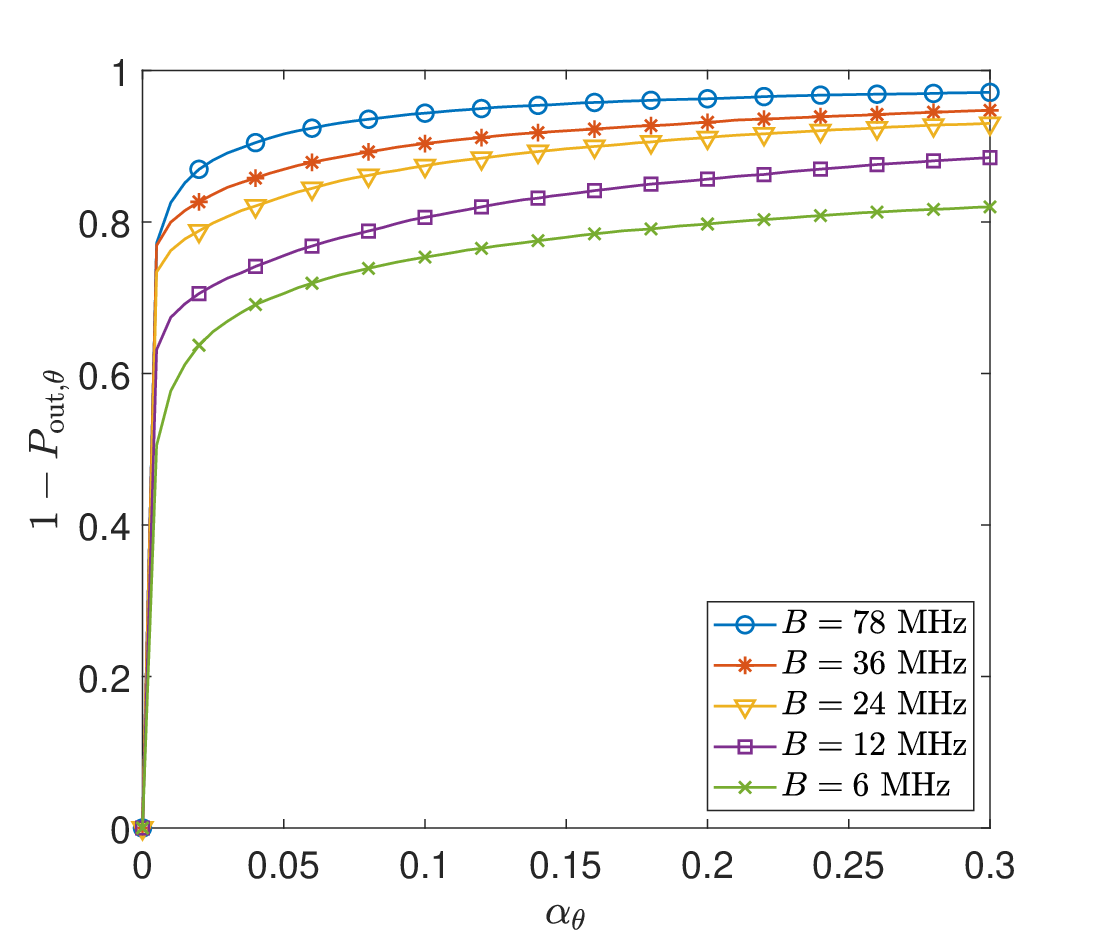}%
  \label{fig:evaluation:avgPrice}%
}
\\
\subfloat[\small{Range estimation for $K=3$ drones}]{%
  \includegraphics[width=8cm]{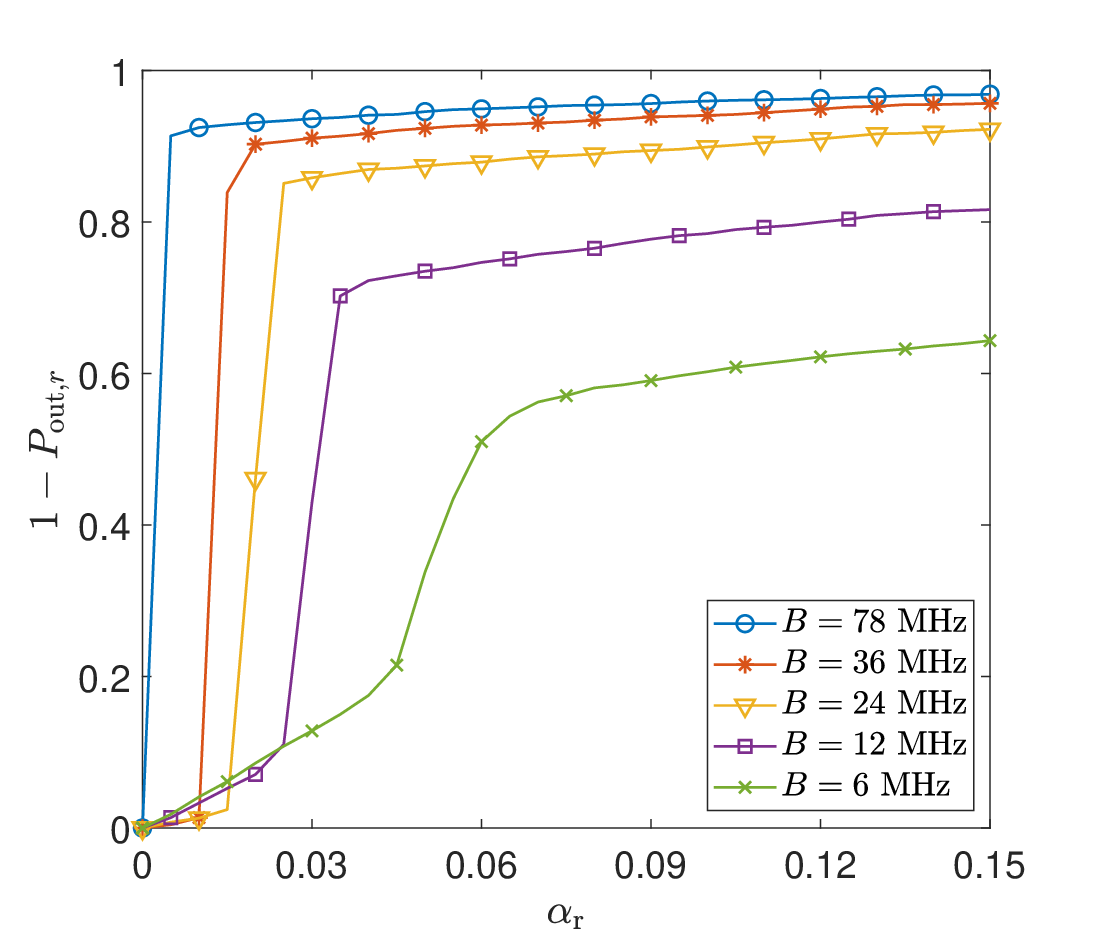}%
  \label{fig:evaluation:revenue}%
}\qquad
\subfloat[\small{\ac{po} estimation for $K=3$ drones}]{%
  \includegraphics[width=8cm]{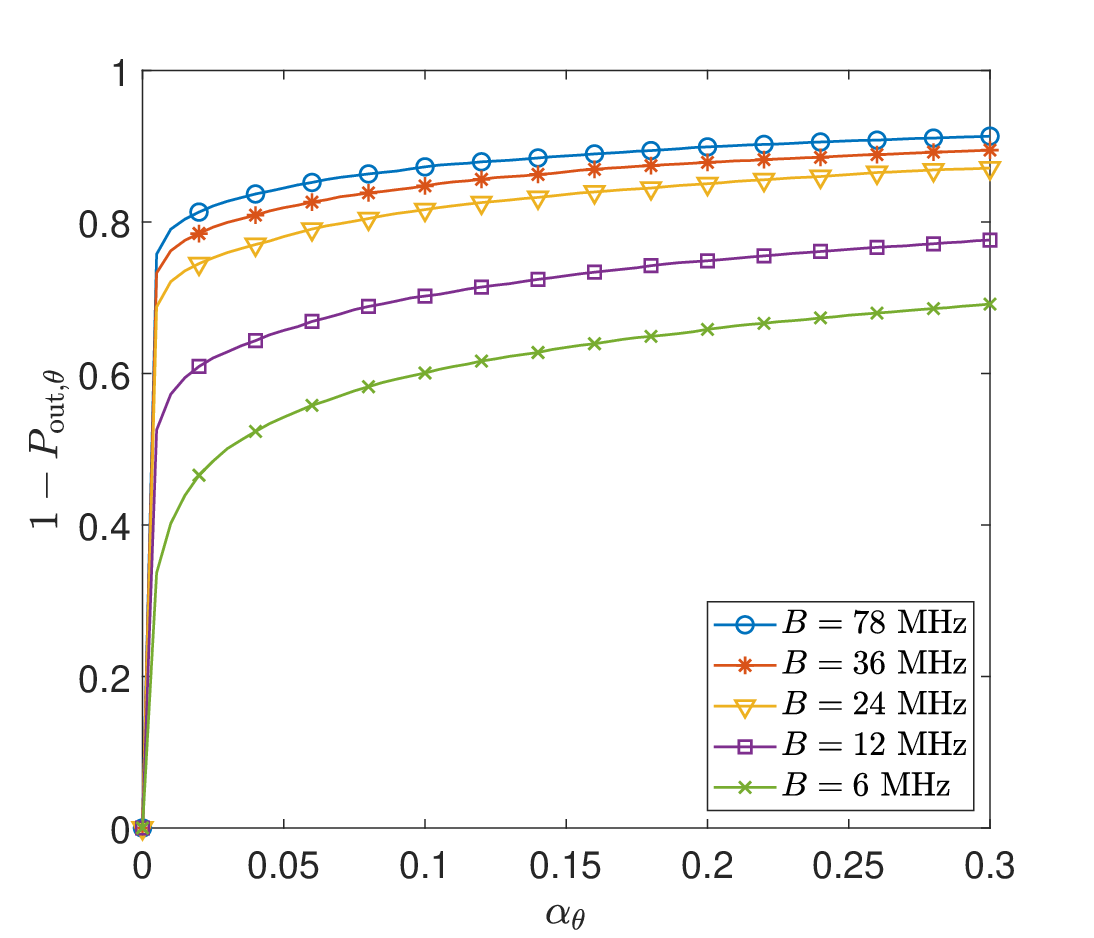}%
  \label{fig:evaluation:avgPrice}%
}
\\
\subfloat[\small{Range estimation for $K=4$ drones}]{%
  \includegraphics[width=8cm]{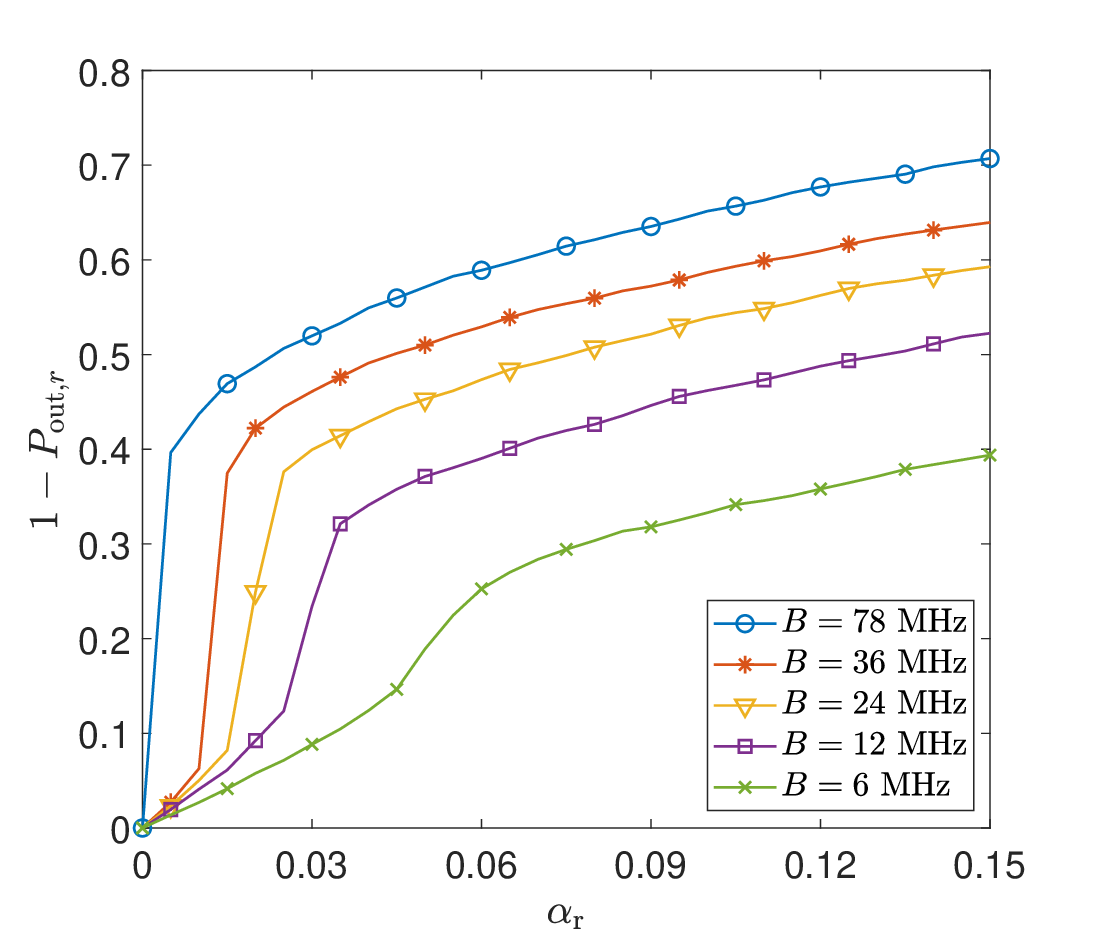}%
  \label{fig:evaluation:revenue}%
}\qquad
\subfloat[\small{\ac{po} estimation for $K=4$ drones}]{%
  \includegraphics[width=8cm]{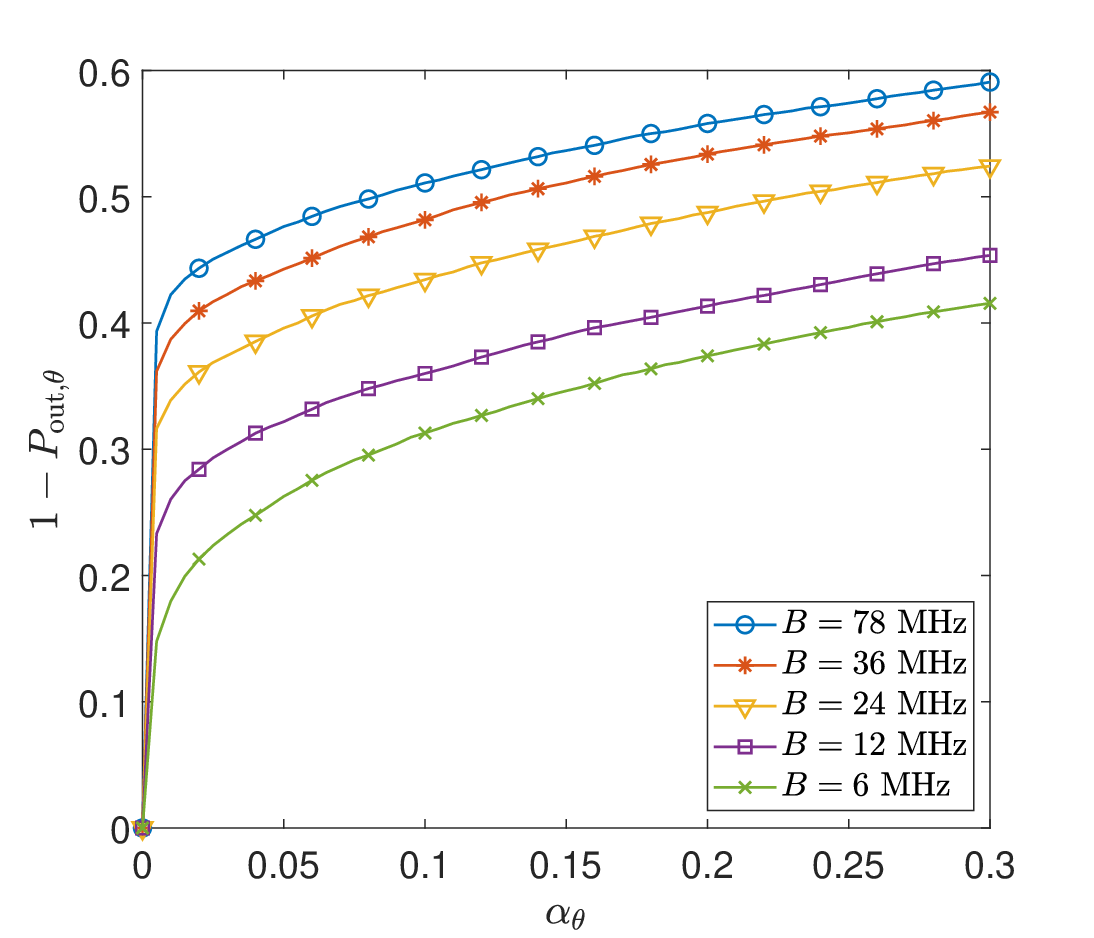}%
  \label{fig:evaluation:avgPrice}%
}
\\
\caption{The performance of the proposed joint \ac{em}-based estimator  for different number of drones, $K$,
 $P_1=P_2=\ldots=P_K=51$ dBm, $r_1,r_2,\ldots r_K\in \mathcal{U}_{\rm c}[1,10]$ Km,
and $N_{\rm r}=5$.}\label{fig:nioio}
\end{figure*}

{Fig.~\ref{fig:nioio} illustrates $1-P_{{\rm{out}},r}$ and $1-P_{{\rm{out}},\theta}$ of the proposed \ac{em}-based joint estimator versus $\alpha_{\rm r}$ and $\alpha_\theta$ for different number of drones $K=2,3,4$ and different values of receive filter bandwidth $B$. As seen, by increasing $B$, the performance of the
\ac{em}-based joint estimator improves because sharp pulses are received at the receiver, and thus; the approximation error of received signal model in \eqref{Observation} reduces.
Moreover, as expected, the larger number of drones, higher estimation error.
The reason is that the number of \ac{gm} components to be estimated by the \ac{em} algorithm exponentially increases with $K$, i.e., $2^K$; however, the number of observation samples remains fixed.}

\begin{figure}[h!]
\vspace{-3em}
\centering
\subfloat[{\small{The effect of $N_{\rm r}$ on range estimation}}]{%
  \includegraphics[width=8cm]{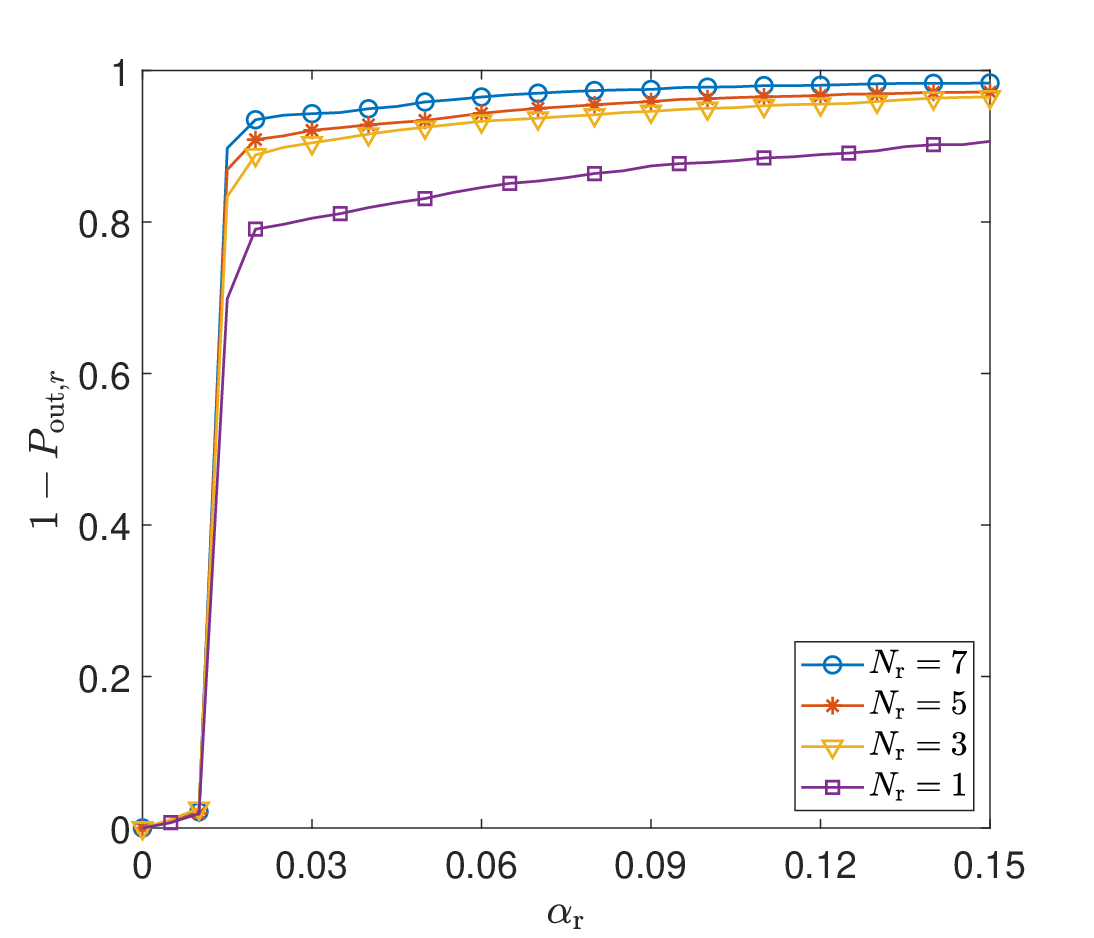}%
  \label{fig:evaluation:revenue1}%
}\\
\subfloat[\small{The effect of $N_{\rm r}$ on \ac{po} estimation}]{%
  \includegraphics[width=8cm]{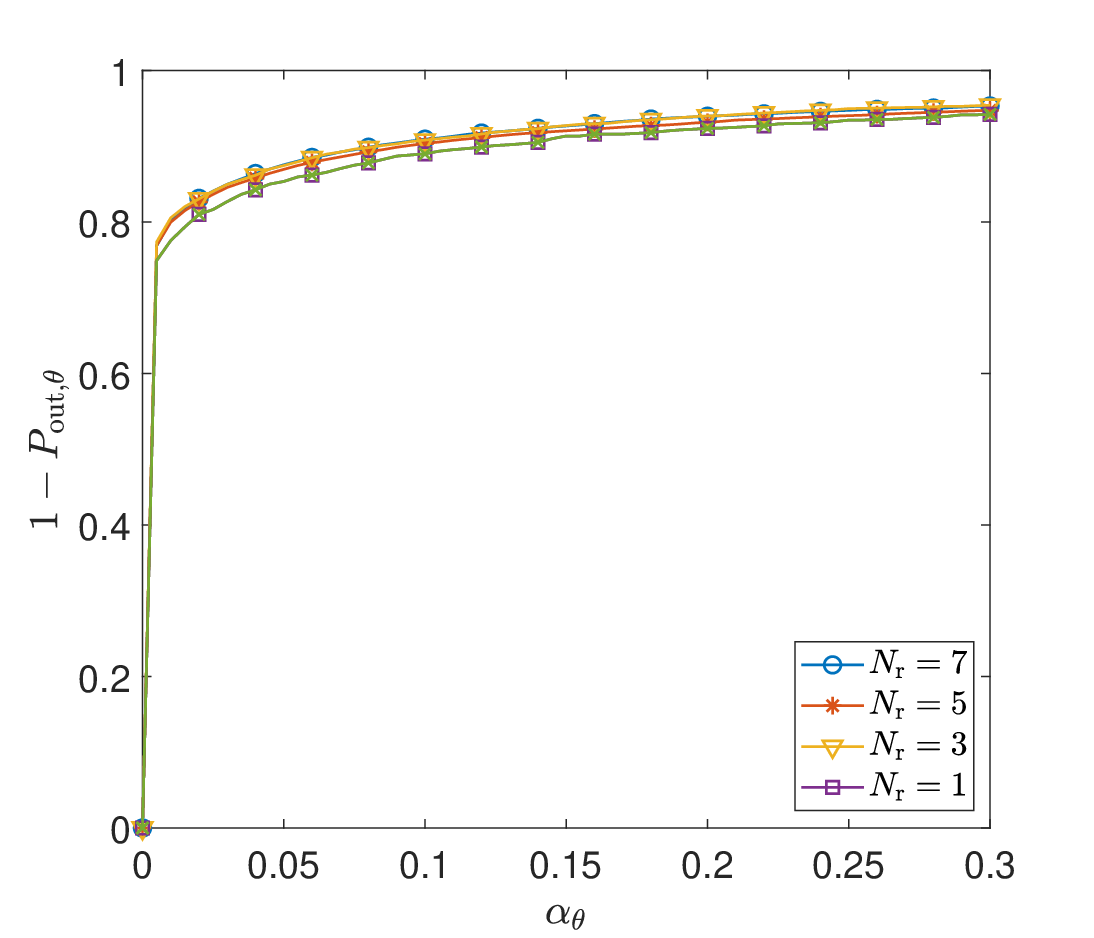}%
  \label{fig:evaluation:avgPrice2}%
}
\caption{{The effect of the number of receive antennas, $N_{\rm r}$, on the performance of the proposed \ac{em}-based joint estimator
for $K=2$ drones, $P_1=P_2=51$ dBm, $r_1,r_2\in \mathcal{U}_{\rm c}[1,10]$ Km, and $\theta_{1,1},\theta_{1,2},\theta_{2,1},\theta_{2,2}, \in \mathcal{U}_{\rm c}[-\pi,\pi)$
.}}\label{Antenna_effect}
\end{figure}

\begin{figure}[h!]
\vspace{-3em}
\centering
\subfloat[{\small{The effect of minimum overlapping on range estimation}}]{%
  \includegraphics[width=8cm]{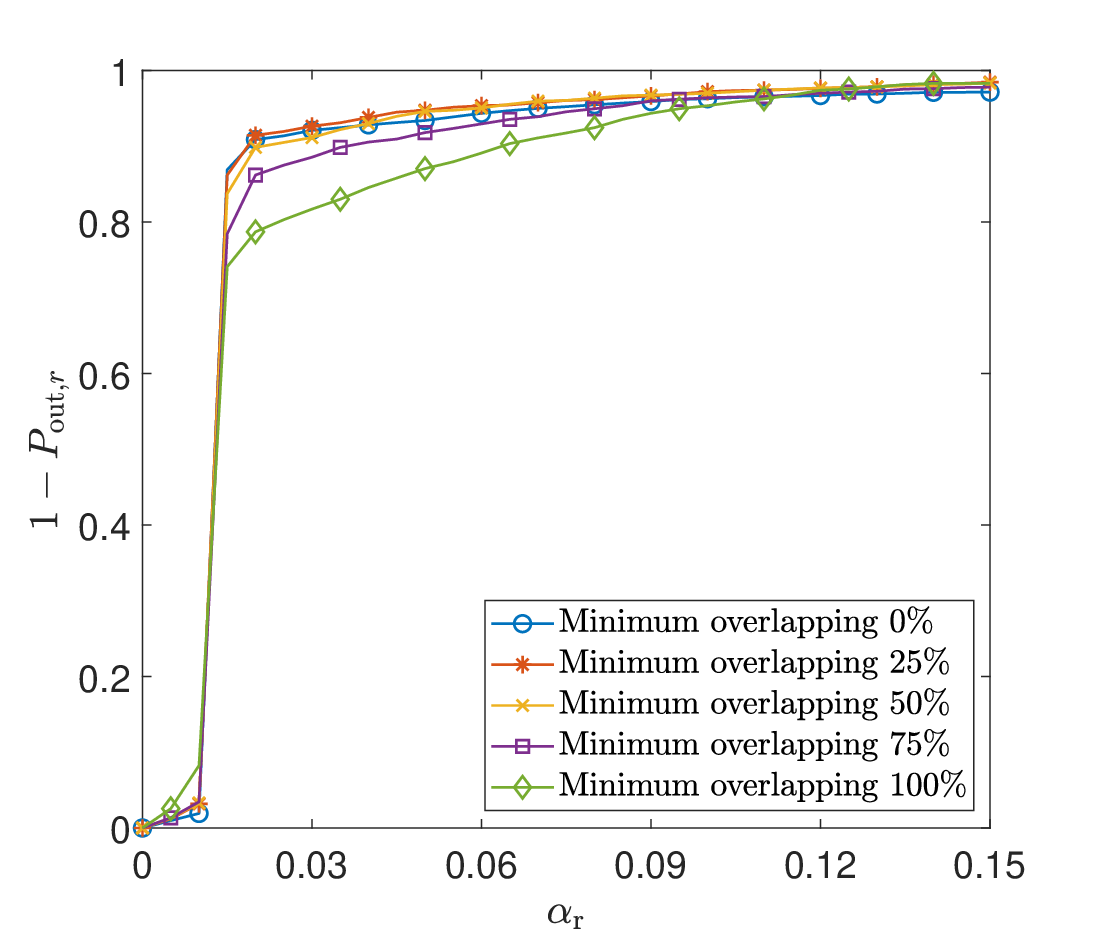}%
  \label{fig:evaluation:revenue1}%
}\\
\subfloat[\small{The effect of minimum overlapping on \ac{po} estimation}]{%
  \includegraphics[width=8cm]{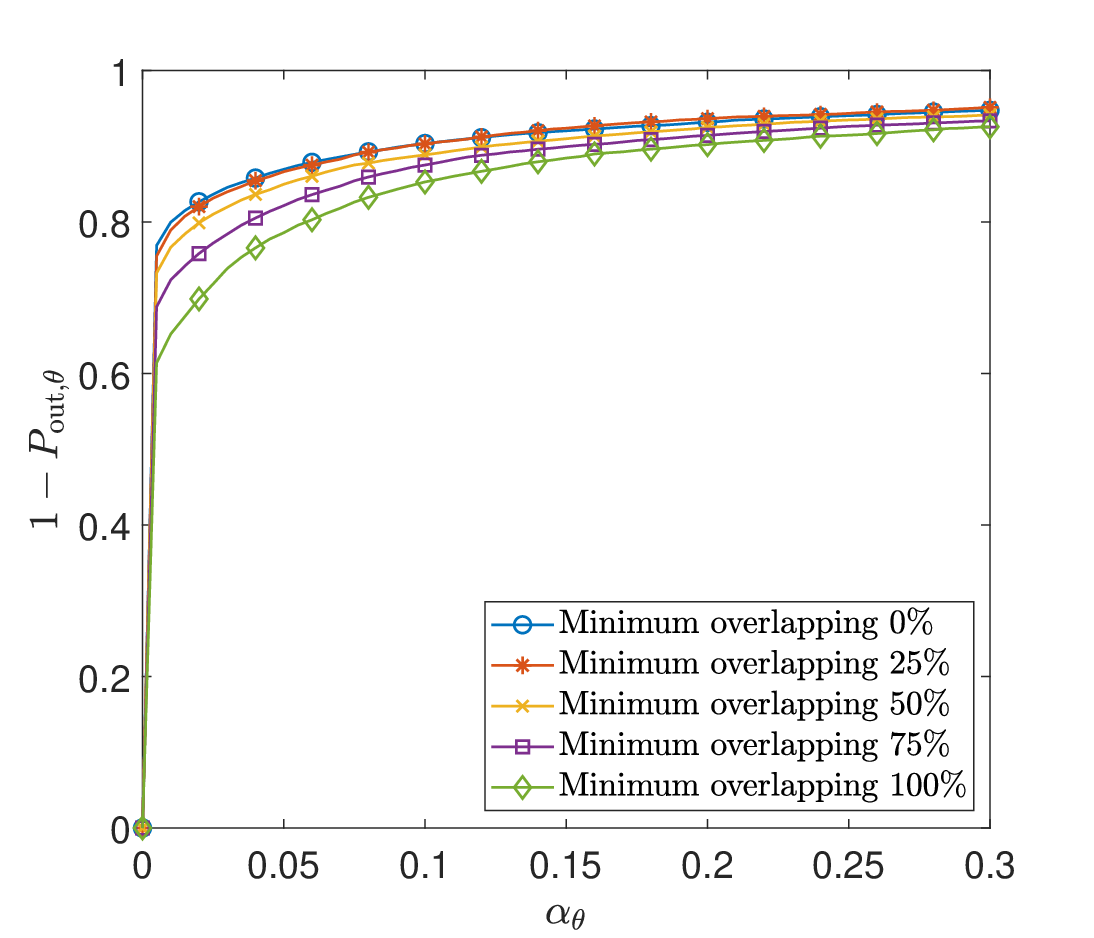}%
  \label{fig:evaluation:avgPrice2}%
}
\caption{{The effect of the minimum overlapping percentage on the performance of the proposed \ac{em}-based joint estimator for $K=2$ drones, $P_1=P_2=51$ dBm, $r_1,r_2\in \mathcal{U}_{\rm c}[1,10]$ Km, $\theta_{1,1},\theta_{1,2},\theta_{2,1},\theta_{2,2}, \in \mathcal{U}_{\rm c}[-\pi,\pi)$, and $N_{\rm r}=5$
.}}\label{Overlapping}
\end{figure}


{Fig.~\ref{Antenna_effect} shows the effect of the number of receive antennas on the range and \ac{po} estimation versus $\alpha_{\rm r}$
and $\alpha_{\rm \theta}$
for $K=2$ drones,
$P_1=P_2=51$ dBm, and $B=36$ MHz.
As seen, the larger number of receive antennas, $N_{\rm r}$, higher  $1-P_{{\rm{out}},r}$.
Also, the rate of performance improvement decreases as the number of receive antennas increases.}
As expected,
increasing the number of receive antennas does not affect $1-P_{{\rm{out}},\theta}$ because spatial diversity
is not employed for the estimation of the $KN_{\rm r}$ independent \ac{po}s. However, the range estimation takes the advantage of spatial diversity.
The interesting property of the proposed algorithm is that it can jointly estimate the range and \ac{po} of multiple drones/aircrafts with a single receive antenna. With a single receive antenna, the range accuracy of $\alpha_r=0.03$ is archived for $80\%$ of the time.

\begin{figure}[h!]
\centering
\subfloat[{Performance comparison for range estimation}]{%
  \includegraphics[width=8cm]{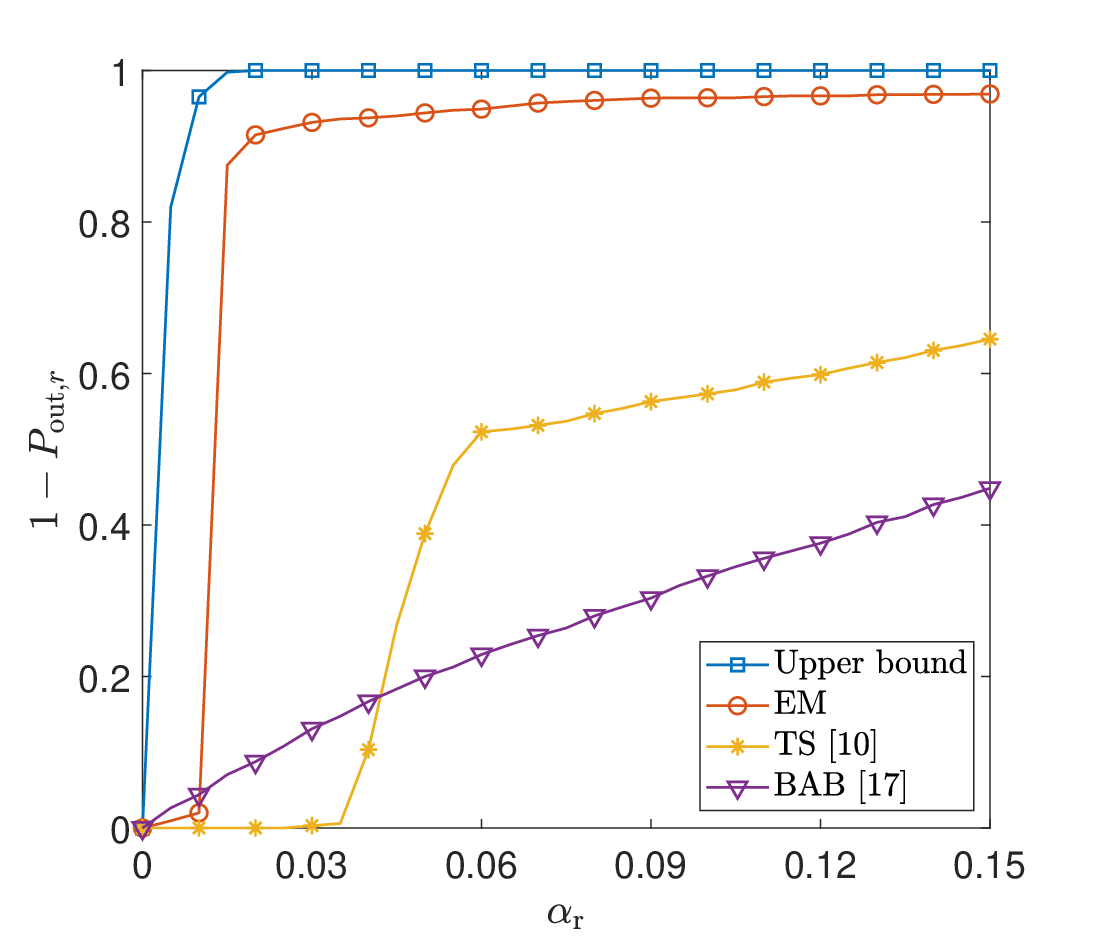}%
  \label{Compare_1}%
}\\
\subfloat[\small{Performance comparison for \ac{po} estimation}]{%
  \includegraphics[width=8cm]{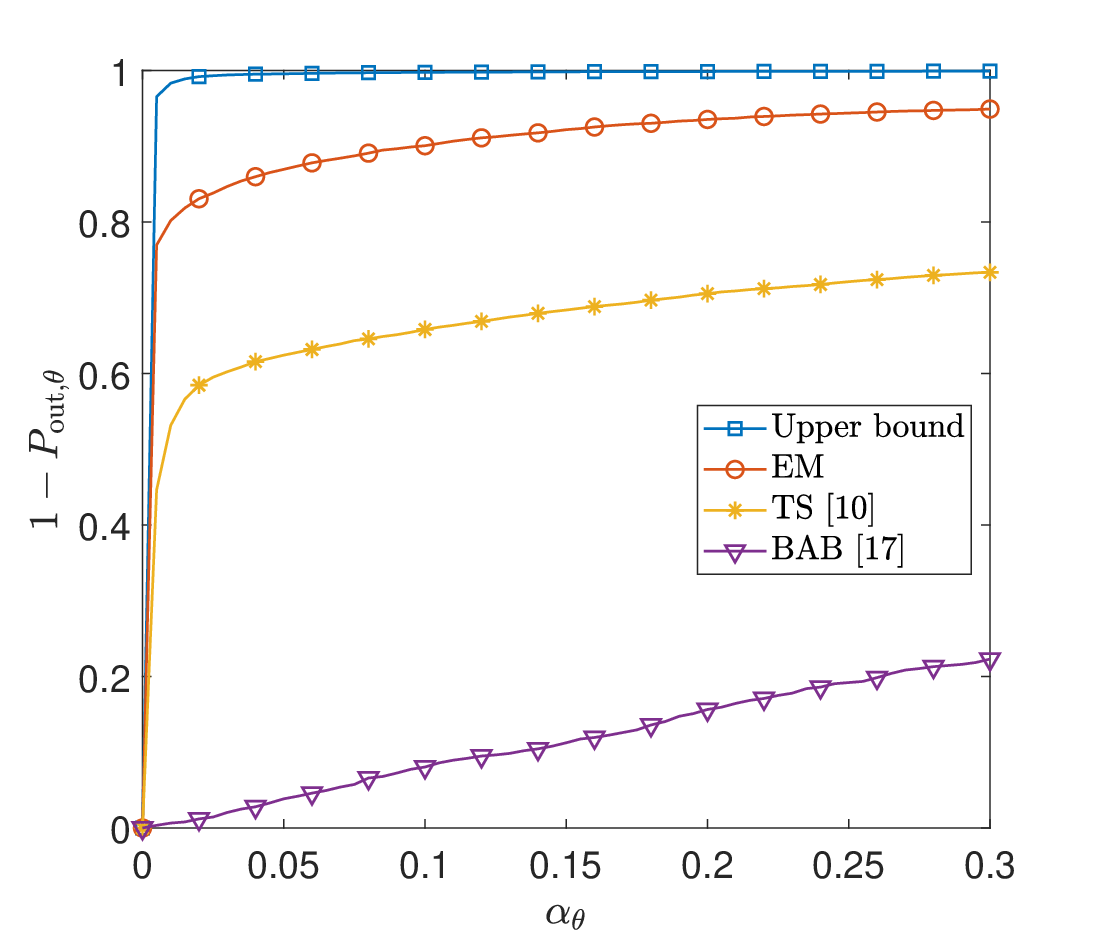}%
  \label{Compare_2}%
}
\caption{{Performance comparison of the proposed \ac{em}-based joint estimator for  $K=2$ drones, $P_1=P_2=51$ dBm, $r_1,r_2\in \mathcal{U}_{\rm c}[1,10]$ Km, $\theta_{1,1},\theta_{1,2},\theta_{2,1},\theta_{2,2}, \in \mathcal{U}_{\rm c}[-\pi,\pi)$, and $N_{\rm r}=5$.
 .}}\label{Compare}
\end{figure}

Fig.~\ref{Overlapping} illustrates the performance of the proposed
\ac{em}-based joint estimator for different minimum
overlapping percentage of the ADS-B packets for $K=2$ drones/aircrafts, $N_{\rm r}=5$ receive antennas, and $B=36$ MHz. As expected,
the lower minimum
overlapping percentage, more accurate ranging and \ac{po} estimation because
interference decreases.

In Fig.~\ref{Compare}, we compare the performance of the proposed  \ac{em}-based  joint estimator with
the time segmentation (TS) and the blind adaptive beamforming (BAB)
ADS-B packet separation methods
in \cite{li2021reliable} and \cite{wang2019ads} for  $K=2$ drones/aircrafts, $N_{\rm r}=5$ receive antennas, and $B=36$ MHz.
The TS and BAB methods first separate the ADS-B packets. Then, by using the separated packets, they can estimate the range and \ac{po} of the drones.
As seen, our proposed method outperforms the TS and the BAB methods because it employs all the observation samples including
the overlapping snapshot
 for ranging and \ac{po} estimation; however, the TS and the BAB methods
 rely on the non-overlapping snapshot for ADS-B packet recovery. Hence, as the delay between the reception of two ADS-B packets decreases, their performance degrades.
It should be mentioned that while the TS and the BAB methods can be used to estimate the range of maximum $K=2$ drones; our proposed method can estimate the range of $K>1$ drones with a single receive antenna.
In Fig. \ref{Compare}, we also show performance of the efficient estimator\footnote{An efficient estimator is an unbiased estimator that attains the  Cramer-Rao Lower Bound (CRLB) \cite{wang2009cramer}.} \cite{kay1993fundamentals}. For the efficient estimator, the transmit symbols and delay of the drones/aircrafts are assumed to be known a priori at the receiver.
As seen, there is a small gap between the performance of the proposed estimator and the efficient estimator while the transmit symbols and
delay of the drones/aircraft are unknown to the \ac{em}-based estimator.

\section{Conclusion}\label{conclusion}
In this paper, we showed that the lost ADS-B packets due to packet collisions can be employed to jointly estimate range and \ac{po} of multiple drones/aircrafts in the airspace. This enables drones to maintain safe operation distance in the congested airspace, where
ADS-B packet decoding is impossible due to packet collisions.
To achieve this, we derived the maximum likelihood and \ac{em}-based joint range and \ac{po} estimators using an approximate \ac{pdf} obtained by \ac{kld} minimization. The proposed estimators consider uncoordinated and asynchronous ADS-B packet transmission. For joint range and \ac{po} estimation, a priori knowledge or estimation of the drones' time delay is not required.
Simulation results showed that the \ac{em}-based joint estimator can estimate the range of multiple drones/aircrafts with a single receive antenna. Performance improvement by employing multiple receive antennas is obtained.

\appendices
\section{Proof of Theorem 1}\label{Appendix_1}
The vector ${\bf x}_{k}$ given hypothesis $H_m^{k}$ in \eqref{As}
is composed of
$4$ ones and $12$ zeros  in ${\bf s}$. The number of zeros and ones in the data field ${\bf d}_{k}$ is $112$ due to Manchester encoding.  Moreover, $M$ zeros are added irrespective to the hypothesis $H_m^{k}$ because of the maximum integer delay. Hence, the total number of zeros and ones in ${\bf x}_{k}$ are ${M+124}$ and ${116}$, respectively.

Let $\mathcal{X}_m^{k}$ denote all possible vector for ${\bf x}_{k}$ given hypothesis $H_m^{k}$.
The cardinality of $\mathcal{X}_m^{k}$, i.e., $|\mathcal{X}_m^{k}|=2^{112}$ since the randomness in
${\bf x}_{k}$ because of the data field ${\bf d}_{k}$ results in $2^{112}$ different Manchester encoded sequences.
Thus; the joint \ac{pmf} of $x_{k,0}, x_{k,1}, \cdots, x_{k,N}$
is given by
\begin{align}\label{Conditional}
f({\bf x}_k)=f\big{(}& x_{k,0}, x_{k,1}, \ldots, x_{k,N}|H_m^{k}\big{)} \\ \nonumber
&=
\begin{cases}
\frac{1}{2^{112}}  \,\,\,\,\,\,\,\,\ {\bf x}_{k} \in \mathcal{X}_m^{k},\\
0 \,\,\,\,\,\,\,\,\,\,\,\,\,\,\,\,\ {\bf x}_{k} \notin \mathcal{X}_m^{k}.
\end{cases}
\end{align}
Let us consider the following approximation for the joint \ac{pmf} in \eqref{Conditional}.
\begin{align}\label{pdf_approximation}
f({\bf x}_k)=f\big{(} x_{k,0}, x_{k,1}, & \ldots, x_{k,N}|H_m^{k}\big{)} \\ \nonumber
& \approx \prod_{n=0}^{N}  g(x_{k,n}|H_m^{k})
\end{align}
The \ac{kld} for $f$ and $g$ in \eqref{pdf_approximation} is given by
\begin{align} \nonumber
D&(f\|g)  = \sum_{\mathcal{X}_m^{k}} \bigg{[} f\big{(} x_{k,0}, x_{k,1}, \ldots, x_{k,N}|H_m^{k}\big{)} \\ \label{KLD}
& \,\,\,\ \times  \ln   \frac{f\big{(} x_{k,0}, x_{k,1}, \cdots, x_{k,N}|H_m^{k}\big{)}}{\prod_{n=0}^{N}  g(x_{k,n}|H_m^{k})}  \bigg{]}.
\end{align}
Because $x_{k,n} \in \{0,1\}$ for $n=0,1,\ldots,N$, the \ac{kld} is minimized for Bernoulli distribution. The \ac{pmf} of
Bernoulli distribution is expressed as
\begin{align}\label{Bernoulli}
g(i|H_m^{k}) = \begin{cases}
   p & \text{if }i=0, \\
    1-p & \text {if } i = 1,
 \end{cases}
\end{align}
where $p \in [0,1]$.

By substituting the \ac{pmf} in \eqref{Conditional} and \eqref{Bernoulli} into \eqref{KLD}, and by using the fact that $\sum_{n=0}^{N}x_{k,n}=116$
for ${\bf x}_{k} \in \mathcal{X}_m^{k}$, we can write the \ac{kld} in \eqref{KLD} as
\begin{align}
D(f\|g) & = \frac{1}{2^{112}} \sum_{\mathcal{X}_m^{k}} \ln \bigg{[} \frac{1}{2^{212} p^{M+124}(1-p)^{116}} \bigg{]} \\ \nonumber
& = - \ln \Big{[}{2^{212} p^{M+124}(1-p)^{116}} \Big{]}.
\end{align}
To obtain $p$, we need to minimize $D(f\|g)$.
Because $\ln(\cdot)$ is a monotonically increasing function,  $D(f\|g)$ is minimized when $p^{M+124}(1-p)^{116}$, $p \in [0,1]$,  is maximized.
By taking the derivative with respect to $p$ and setting it to zero, we obtain
\begin{align}\label{Derivative}
\diff{\big{(}p^{M+124}(1-p)^{116}\big{)}}{p} &={(M+124)}p^{M+123}(1-p)^{116} \\ \nonumber
&-116 (1-p)^{115}p^{M+124}=0.
\end{align}
By solving \eqref{Derivative}, we obtain \eqref{Bernoilli_Parameter}.

\section{Proof of the \ac{pdf} of $G=\sum_{k=1}^{K} Z_k$}\label{Appx_B}
Let us consider $G=\sum_{k=1}^{K} Z_k$. Since $Z_k$, $k=1,2,\cdots,K$, are independent random variables, the \ac{pdf} of $G$ is given by
\begin{align}\label{pdf_convolution}
 f_{G}(g;p,{\bf{h}}) & = f_{{Z}_{1}}\big{(}z;p,h_1\big{)} * f_{{Z}_{2}}\big{(}z;p,h_2\big{)} \\ \nonumber
& \,\,\,\,\,\ * \ldots  *
f_{{Z}_{K}}\big{(}z;p,h_K\big{)},
\end{align}
where $z \triangleq z_{\rm r}+iz_{\rm I} \in \mathbb{C}$, ${\bf{h}}\triangleq [h_1,h_2,\cdots,h_K]^T$, $h_k \triangleq {(h_k)}_{\rm r}+i{(h_k)}_{\rm I}  \in \mathbb{C}$,
and $f_{{Z}_{k}}\big{(}z;p,h_k\big{)}$ is given in \eqref{Theory_1}.
The two-dimensional Laplace transform of $f_{{Z}_{k}}\big{(}z;p,h_k\big{)}$ is expressed as
\begin{align}\nonumber
&F_{Z_k}\big{(}s_1,s_2;p,h_k\big{)} = \int_{-\infty}^{+\infty} \hspace{-0.2em} \int_{-\infty}^{+\infty}\hspace{-0.5em} f_{{Z}_{k}}\big{(}z;p,h_k\big{)} e^{-s_1{z_{\rm r}}-s_2 {z}_{\rm I}}{\rm d}{z}_{\rm r} {\rm d}z_{\rm I}  \\ \nonumber
&=\int_{-\infty}^{+\infty}  \int_{-\infty}^{+\infty}  p \delta(z_{\rm r})\delta({z}_{\rm I}) e^{-s_1 {z_{\rm r}}-s_2 {z_{\rm I}}}{\rm d}{z}_{\rm r} {\rm d}z_{\rm I} \\ \nonumber
& +\int_{-\infty}^{+\infty} \hspace{-0.5em} \int_{-\infty}^{+\infty} \hspace{-0.5em} (1-p) \delta({z_{\rm r}}-{(h_k)}_{\rm r}) \delta(z_{\rm I}-{(h_k)}_{\rm I})
 e^{-s_1 {z_{\rm r}}-s_2 {z_{\rm I}}}{\rm d}{z}_{\rm r} {\rm d}z_{\rm I} \\ \nonumber
&= p + (1-p)  e^{-({(h_{k)}}_{\rm r} s_1+{(h_{k})}_{\rm I}s_2)}.
\end{align}
Using the fact that the linear convolution is equivalent to multiplication in the Laplace domain, we can write
\begin{align}
F_{G}\big{(}s_1,s_2;p,{\bf h}\big{)} & = \prod_{k=1}^{K}  F_{Z_k}\big{(}s_1,s_2;p,h_k\big{)}  \\ \nonumber
& = \prod_{k=1}^{K} \Big{[} p + (1-p) e^{-({(h_k)}_{\rm r}s_1+{(h_k)}_{\rm I}s_2)} \Big{]},
\end{align}
where $F_{G}\big{(}s_1,s_2;p,{\bf h}\big{)}$ denotes the two-dimensional Laplace transform of $f_{G}(g;p,{\bf{h}})$.
Let us consider the multi-binomial theorem as follows \cite{morris1975central}
\begin{align}\label{multi-binomia}
&\prod_{i=1}^{d} (a_i+b_i)^{n_i} \\ \nonumber
&= \sum_{v_1=0}^{n_1} \cdots \sum_{v_d=0}^{n_d} \binom{n_1}{v_1} a_1^{v_1} b_1^{n_1-v_1} \cdots \binom{n_d}{v_d} a_d^{v_d} b_d^{n_d-v_d}.
\end{align}
By substituting $a_i \triangleq  p$, $b_i  \triangleq  (1-p) e^{-({(h_i)}_{\rm r}s_1+{(h_i)}_{\rm I}s_2)}$, $d=K$, and $n_i=1$, $i=1,2,\ldots,d$, into \eqref{multi-binomia}, we obtain
\begin{align} \nonumber
F_{G}& \big{(}s_1,s_2;p,{\bf h}\big{)}  = \sum_{v_1=0}^{1} \cdots \sum_{v_K=0}^{1}
\bigg{[} p^{\sum_{k=1}^{K} v_k} (1-p)^{K-{\sum_{k=1}^{K} v_k}} \\ \label{Lalace2}
& \times \exp{\Big{[}-\sum_{k=1}^{K}(1-v_{k})({({{h}}_{k})}_{\rm r} s_1+{(h_k)}_{\rm I}s_2)\Big{]}} \bigg{]}.
\end{align}
The two-dimensional inverse Laplace transform of the exponential term in \eqref{Lalace2} is given by
\begin{align}\label{Laplac_delta}
&\mathcal{L}^{-1}\bigg{\{}\exp{\Big{[}-\sum_{k=1}^{K}(1-v_{k})({(h_k)}_{\rm r}s_1+{(h_k)}_{\rm I}s_2)\Big{]}} \bigg{\}} \\ \nonumber
& =\delta \Big{(}s_{\rm r}-\sum_{k=1}^{K}(1-v_{k}){(h_k)}_{\rm r} \Big{)} \delta \Big{(}s_{\rm I}-\sum_{k=1}^{K}(1-v_{k}){(h_k)}_{\rm I} \Big{)} \\ \nonumber
&=\delta_{\rm c} \Big{(}s-\sum_{k=1}^{K}(1-v_{k})h_k\Big{)},
\end{align}
where $s=s_{\rm r}+i s_{\rm I} \in {\mathbb C}$.
By taking two-dimensional inverse Laplace transform from both sides of \eqref{Lalace2} and  then employing \eqref{Laplac_delta}, we obtain $f_{G}(g;p,{\bf{h}})$ as in \eqref{delta}.

\section{}\label{Appendix 3}
For $K=2$, we have
\begin{align}
&\bm{\mu} \triangleq [\mu_0 \ \mu_1 \  \mu_2 \ \mu_3]^T  \\ \nonumber
&=[\beta_1 \exp(j\theta_1)+\beta_2 \exp(j\theta_2) \,\,\ \beta_2 \exp(j\theta_2) \,\,\ \beta_1\exp(j\theta_1) \,\,\ 0]^T,
\end{align}
where $\beta_1  > \beta_2$, and for $K=3$, we can write
\begin{align}
\bm{\mu}  & \triangleq [\mu_0 \ \mu_1 \  \mu_2 \ \mu_3 \ \mu_4 \ \mu_5 \  \mu_6 \ \mu_7 ]^T \\ \nonumber
&=\big{[}\beta_1 \exp(j\theta_1)+\beta_2 \exp(j\theta_2)+ \beta_3 \exp(j\theta_3) \\ \nonumber
& \,\,\,\,\,\,\ \beta_2\exp(j\theta_2)+\beta_3\exp(j\theta_3) \,\,\ \beta_1\exp(j\theta_1)+\beta_3\exp(j\theta_3) \\ \nonumber
& \,\,\,\,\,\,\ \beta_3\exp(j\theta_3) \,\,\  \beta_1\exp(j\theta_1)+\beta_2\exp(j\theta_2) \,\,\ \beta_2\exp(j\theta_2) \\ \nonumber
& \,\,\,\,\,\,\ \beta_1\exp(j\theta_1) \,\,\  0 \big{]}^T,
\end{align}
with $\beta_1  > \beta_2  >\beta_3$.
For $K=4$, we can write
\begin{align}\label{4mu}
\bm{\mu}  & \triangleq [\mu_0 \ \mu_1 \  \mu_2 \ \mu_3 \ \mu_4 \ \mu_5 \  \mu_6 \ \mu_7 \ \ldots \ \mu_{13} \  \mu_{14} \ \mu_{15} ]^T \\ \nonumber
&=\big{[}\beta_1 \exp(j\theta_1)+\beta_2 \exp(j\theta_2)+ \beta_3 \exp(j\theta_3) +  \beta_4 \exp(j\theta_4) \\ \nonumber
& \,\,\,\,\,\,\ \beta_2\exp(j\theta_2)+\beta_3\exp(j\theta_3) + \beta_4\exp(j\theta_4) \\ \nonumber
& \,\,\,\,\,\,\ \beta_1\exp(j\theta_1)+\beta_3\exp(j\theta_3) + \beta_4\exp(j\theta_4) \\ \nonumber
& \,\,\,\,\,\,\ \beta_3\exp(j\theta_3)+\beta_4\exp(j\theta_4)\\ \nonumber
& \,\,\,\,\,\,\ \beta_1\exp(j\theta_1)+\beta_2\exp(j\theta_2)+\beta_4\exp(j\theta_4)\\ \nonumber
& \,\,\,\,\,\,\ \beta_2\exp(j\theta_2)+\beta_4\exp(j\theta_4)\\ \nonumber
& \,\,\,\,\,\,\ \beta_1\exp(j\theta_1)+\beta_4\exp(j\theta_4)\,\,\,\,\ \beta_4\exp(j\theta_4) \\ \nonumber
& \,\,\,\,\,\,\ \beta_1\exp(j\theta_1)+\beta_2\exp(j\theta_2) + \beta_3\exp(j\theta_3) \\ \nonumber
& \,\,\,\,\,\,\ \beta_2\exp(j\theta_2)+\beta_3\exp(j\theta_3)  \,\,\,\,\    \beta_1\exp(j\theta_1)+\beta_3\exp(j\theta_3)   \\ \nonumber
& \,\,\,\,\,\,\ \beta_3\exp(j\theta_3) \,\,\,\,\ \beta_1\exp(j\theta_1)+\beta_2\exp(j\theta_2) \,\,\,\,\ \beta_2\exp(j\theta_2) \\ \nonumber
& \,\,\,\,\,\,\ \beta_1\exp(j\theta_1) \,\,\,\,\ 0 \big{]}^T,
\end{align}
where $\beta_1>\beta_2>\beta_3>\beta_4$.

\end{document}